\documentclass[preprint,11pt,3p,authoryear]{elsarticle}
\usepackage{booktabs}
\usepackage{lineno}
\modulolinenumbers[5]
\usepackage{tabularx}
\usepackage{multirow}
\usepackage{pdflscape}
\usepackage{setspace} 
\usepackage[,ruled,vlined,noend,linesnumbered]{algorithm2e} 
\usepackage{amsmath}
\usepackage{amssymb} 
\usepackage{bm}
\usepackage{mathtools}
\usepackage[table]{xcolor}
\usepackage{xspace} 
\usepackage{natbib}
\usepackage{multirow}
\usepackage{float}
\usepackage{setspace}
\usepackage{tikz}
\usepackage{xcolor}
\usepackage{fullpage}
\usepackage[bookmarks=false,colorlinks]{hyperref}
\usepackage[capitalize]{cleveref}
\crefname{problem}{problem}{problems}
\Crefname{problem}{Problem}{Problems}
\usepackage{subcaption}
\usepackage{pgf,tikz,pgfplots}
\pgfplotsset{compat=1.15}
\usepackage{mathrsfs}
\usetikzlibrary{arrows}
\newtheorem{lemma}{Lemma}

\newtheorem{proof}{Proof}
\newtheorem{remark}{Remark}
\newtheorem{definition}{Definition}
\newtheorem{modl}{Model}

\makeatletter
\def\ps@pprintTitle{%
  \let\@oddhead\@empty
  \let\@evenhead\@empty
  \def\@oddfoot{\reset@font\hfil\thepage\hfil}
  \let\@evenfoot\@oddfoot
}
\makeatother

\newcommand{\walks}{\mathcal{W}}

\newcommand{\PC}{$\emph{min}$ $k$-FDPC\xspace}
\newcommand{\PT}{$\emph{min}$ $k$-FDT\xspace}
\newcommand{\PW}{$\emph{min}$ $k$-FDW\xspace}

\newcommand{\newcom}[1]{\color{pink}}

\usepackage{graphicx}

\journal{European Journal of Operational Research}

\begin{document}

\begin{frontmatter}
\title{Minimum Flow Decomposition in Graphs with Cycles\\ using Integer Linear Programming}

\author[UH]{Fernando H. C. Dias}\corref{mycorrespondingauthor}
\ead{fernando.cunhadias@helsinki.fi}
\cortext[mycorrespondingauthor]{Corresponding author}
\author[UM]{Lucia Williams}
\author[UM]{Brendan Mumey}
\author[UH]{Alexandru I. Tomescu}
\address[UH]{Department of Computer Science, University of Helsinki, Finland}
\address[UM]{School of Computing, Montana State University, Bozeman, MT, USA}

\begin{abstract}
Minimum flow decomposition (MFD) --- the problem of finding a minimum set of weighted source-to-sink paths that perfectly decomposes a flow --- is a classical problem in Computer Science, and variants of it are powerful models in a different fields such as Bioinformatics and Transportation. Even on acyclic graphs, the problem is NP-hard, and most practical solutions have been via heuristics or approximations. While there is an extensive body of research on acyclic graphs, currently there is no \emph{exact} solution on graphs with cycles. In this paper we present the first ILP formulation for three natural variants of the MFD problem in graphs with cycles, asking for a decomposition consisting only of weighted source-to-sink paths or cycles, trails, and walks, respectively. On three datasets of increasing levels of complexity from both Bioinformatics and Transportation, our approaches solve any instance in under 10 minutes. 
Our implementations are freely available at \url{github.com/algbio/MFD-ILP}.
\end{abstract}

\begin{keyword}
Integer programming \sep Bioinformatics \sep Network flow \sep Flow decomposition \sep Transportation science
\end{keyword}

\end{frontmatter}

\newpage

\section{Introduction}
\label{into}

\subsection{Background}
\label{sec:background}

Flow decomposition (FD) is a classical and well-researched network problem in which a source-to-sink flow needs to be decomposed into a set of weighted source-to-sink paths (and possibly cycles) such that their respective weights perfectly fit each edge's flow. An essential textbook property (see, e.g.~ \citep{ahuja1988network}) is that any flow in a graph with $m$ edges can be decomposed into at most $m$ paths or cycles (each with some associated weight). However, some flow paths or cycles could be suitably combined
to form a smaller decomposition. Therefore, a popular and practically-motivated variant of the flow decomposition problem is to find the decomposition into the \emph{minimum} amount of paths or cycles (Minimum Flow Decomposition problem, MFD). This variant, however, is NP-hard, even if the input network
is restricted to a directed acyclic graph (DAG), as first observed by~\cite{VATINLEN20081390}.

Despite their NP-hardness, MFD and its variants are used in many applications, such as computer networking
\citep{mumey2015parity,hartman2012split,cohen2014effect,hong2013achieving},
Transportation science \citep{olsen2020study, Ohst:2015aa}, and bioinformatics \citep{dias2022fast,kovaka2019transcriptome,kloster2018practical,scallop2, jumper}. The input graph may naturally be a DAG in some applications, such as reference-based RNA transcript assembly in bioinformatics. However, in many applications, the input network graph naturally contains cycles, as in computer and transportation networks, or sequence assembly problems from bioinformatics \citep{grabherr2011trinity,schulz2012oases}. 

Despite the need to decompose flows with cycles in applications, much of the recent progress on MFD has focused on DAGs, due to properties associated with acyclicity that can be exploited into algorithms, as in the fixed-parameter tractable (FPT) algorithms of \cite{kloster2018practical} and \cite{williams2021flow}, the integer linear programming (ILP) formulation of \cite{dias2022fast}, and the heuristic algorithm called Catfish~\citep{shao2017theory}.

In fact, no exact (that is, returning a truly \emph{minimum} solution) solver exists for MFD on graphs with cycles. Because different applications may require different types of decompositions, in this paper, we study three natural versions of the problem on cyclic inputs: decomposing into source-to-sink paths or cycles (in this paper, paths and cycles do not repeat nodes, except for cycles where the first and last nodes are the same), into source-to-sink trails (which may repeat nodes but not edges), and into source-to-sink walks (which may repeat both nodes and edges).

\subsection{Related Work}

The MFD problem was first formally studied in~\citep{VATINLEN20081390}, where it was shown to be NP-hard,
even on DAGs. This result was later strengthened by~\cite{hartman2012split}, who proved several additional results for the problem. They showed that the problem is also NP-hard when the flow values come from only the set $\{1, 2, 4\}$, that MFD is hard to approximate (i.e., there is some $\epsilon > 0$ such that MFD cannot be approximated within a $(1+\epsilon)$ factor, unless P=NP), even on DAGs. On the positive side, on DAGs,
it is possible to decompose all but a $\varepsilon$-fraction of the flow within a $O(1/\varepsilon)$ factor of the optimal number of paths. \cite{mumey2015parity} gave an approximation algorithm for DAGs with an exponential approximation factor based on decomposing the flow into paths with weights that are powers of two. Recently, \cite{Caceres:2022vi} showed that if the weights of the flow decomposition paths can also take negative integer values, then a polynomial-time approximation algorithm exists for DAGs, with an approximation factor of $\lceil\log||f||\rceil + 1$, where $||f||$ is the largest flow value of any edge. Additionally, \cite{kloster2018practical} showed that the problem on DAGs is FPT in the size of the
minimum decomposition.

In practice, though, applications tend to use heuristics, such as the greedy methods based on choosing the widest or longest paths~\citep{VATINLEN20081390}, and these heuristics can be applied to both DAGs and general graphs, possibly with cycles. Some of these greedy methods could be improved on DAGs by making iterative modifications to the flow graph before finding a greedy decomposition, as shown by \cite{shao2017theory}. However, \cite{Caceres:2022vi} also showed that there exist some instances (even DAGs) where the MFD has size $O(\log m)$ ($m$ is the number of edges of the graph), where such greedy methods return flow decompositions as large as $\Omega(m/\log m)$, meaning that they can be exponentially worse than the optimum.

Many application-oriented algorithms on DAGs for MFD and related problems use ILP, taking advantage of existing software such as Gurobi \citep{gurobi} and CPLEX~\citep{studio2017cplex}. However, most of these solutions encode every possible source-to-sink path as a variable, yielding exponential-size ILP formulations that are impractical to solve for larger instances, which occur in applications such as computer networking~\citep{cohen2014effect}.
A common strategy amongst bioinformatics applications to deal with this issue is to pre-select some subset of the possible paths in the graph, either for all instances (as in vg-flow~\citep{baaijens2020strain} and CLIIQ~\citep{lin2012cliiq}), or only when the input is extensive
(as in MultiTrans~\citep{zhao2021multitrans} and SSP~\citep{safikhani2013ssp}). However, by pre-selecting only some paths that can be part of a solution,
these algorithms may return a non-optimal decomposition.

Recently, two polynomial-size ILP formulations have been proposed that can be used to solve MFD on DAGs without pre-selecting paths. Namely, the program JUMPER \citep{jumper} requires an additional condition on the input DAG, namely that it has a unique Hamiltonian path. This allows all paths to be uniquely determined by subsets of edges that do not pairwise overlap along the Hamiltonian path, yielding a formulation for decomposing into $k$ paths using only a quadratic number of variables and constraints. In~\cite{dias2022fast}, the authors give a general solution working on any DAG to find a size $k$ decomposition, even without a Hamiltonian path, again using only a quadratic number of variables and a linear number of constraints. There, the insight is that source-to-sink paths can be encoded using conservation of flow constraints, as in~\cite{taccari2016integer}. However, neither of these approaches can be straightforwardly extended to handle non-DAG inputs.

As outlined above, most of the contributions to the flow decomposition problem have been restricted to DAGs, the only attempts for the general cyclic case that we know of being limited to heuristics. Since these heuristics can be exponentially worse than the optimum (as discussed above), the quest for efficient and exact MFD approaches to the cyclic case is wide open.

\subsection{Our contributions}

In this paper, we give the first \emph{exact} solutions for MFD on graphs with cycles, based on ILP, for all three natural variants mentioned in \Cref{sec:background}. In addition, we show that our solutions are also efficient through experiments on three datasets from two application domains, Bioinformatics and Transportation Science.

While our solutions follow the same high-level approach from \citep{dias2022fast} when solving MFD on DAGs (i.e., first model each of the $k$ paths by an auxiliary unit flow and then requires that these paths admit weights so that they form a flow decomposition), the cyclic case poses several difficulties, which we overcome as follows. 

As a major novelty, in this paper, we show how to formulate different types of walks in graphs with cycles (i.e., paths \emph{or} cycles, trails, walks) using ILP. The simple path-modelling techniques from~\citep{dias2022fast} heavily rely on the acyclicity of the graph; thus, here, we need to develop new techniques. Another difficulty we overcome is that trails can visit nodes multiple times, and walk can visit both nodes and edges multiple times. Apart from complicating the modelling, these facts also render some constraints more challenging to linearize than in~\citep{dias2022fast}.
Our solutions can be summarized as follows:
\begin{itemize}
\item For the paths \emph{or} cycles variant, we extend the so-called ``sequential formulation'' of~\cite{miller1960integer} (see also \citep{taccari2016integer}) to uniformly model either a path \emph{or} a cycle (i.e., by a \emph{single} set of variables with associated constraints). We still model them by an auxiliary unit flow but add new constraints to require that this auxiliary flow induces either a path \emph{or} a cycle. 
\item For the trails and walks variants, we give a novel reachability-based formulation characterizing when the auxiliary unit flow corresponds to a trail or to a walk, as follows. In order to model a source-to-sink walk, more than simply modelling a unit flow is required because in graphs with cycles, and it may induce \emph{isolated} strongly connected components. Thus, we add constraints to require that any node incident to the unit flow modelling the walk is reachable from the source via a spanning tree rooted at the source and using edges of the unit flow. 
\item For the trails variant, we also propose an alternative approach that forbids such isolated strongly connected components via an iterative constraint generation procedure. This approach is inspired by previous constraint generation approaches from \citep{dantzig1954solution,taccari2016integer}. However, these previous approaches model paths by forbidding only cycles (since a unit flow induces a path if and only if it contains no cycles). Trails and walks can contain cycles, so we need to forbid entire strongly connected components.
\end{itemize}

All the formulations are of quadratic size in the number of constraints and variables, except for the constraint generation approach, which introduces several linear constraints in the graph's number of edges per iteration.

Our ILPs constitute the \emph{first} exact solutions for the NP-hard MFD problem on graphs with cycles for its three natural variants on such inputs. At the same time, we also show that they are efficient with biological and transportation data from different sources, solving simple instances in under 60 seconds and more complex ones in under {12 minutes}. Lastly, mathematical programming can provide significant flexibility that is not easy to incorporate in purely algorithmic formulations. In fact, our ILP formulations can be easily extended to practical aspects of applications. This can be done either by suitably adapting the formulation of what constitutes an element of a decomposition (here, we extended paths to paths or cycles, trails, and walks) or by adapting the condition of when the weighted paths fit the flow edges (for example, by considering inexact, or imperfect flow decomposition as in \citep{williams2019rna} and \citep{dias2022fast}).

The paper is organized as follows. We review the basic concepts of flow networks and flow decomposition and define the problems addressed in \Cref{prelim}. We next present the ILP models for each problem variant, starting with paths or cycles in \Cref{pc}, trails in \Cref{pt}, and walks in \Cref{pw}. Numerical experiments are provided in \Cref{exp}, and concluding remarks are discussed in \Cref{con}.

\section{Preliminaries}
\label{prelim}

\subsection{Basic notions}

In this paper, by a graph $G = (V,E)$, we mean a directed graph with $E \subseteq V \times V$. We also assume that all graphs are weakly connected in the sense that there is an undirected path between any pair of nodes. A \emph{source} of $G$ is a node without incoming edges, and a \emph{sink} is a node without outgoing edges. We will use $n$ and $m$ to denote the cardinality of the sets of nodes V and edges E, respectively, of the graph.

A \emph{walk} in $G$ in a sequence of edges in $E$ where consecutive edges share an end and a start. Given a walk $W$ in $G$, and an edge $(u,v)$ of $G$, we denote by $W(u,v)$ the number of times $W$ passes through $(u,v)$ (and thus 0 if $W$ does not pass through $(u,v)$). We say that $W$ is a \emph{trail} if $W(u,v) = 1$ for all edges $(u,v)$ in $W$. We say that $W$ is a \emph{path} if $W$ repeats no node, and that $W$ is a \emph{cycle}  if $W$ repeats no node, with the exception that the first and the last node are the same. If $W$ has first node $s$ and last node $t$, we say that $W$ is an \emph{$s$-$t$ walk} (or \emph{$s$-$t$ trail}, or \emph{$s$-$t$ path}, accordingly). A \emph{strongly connected component of $G$} is an inclusion-maximal set $C$ of nodes of $G$ such that for any $x,y \in C$, there is an $x$-$y$ path, and an $y$-$x$ path.

\subsection{Network flows and flow decompositions}
\label{sec:network-flows-and-decompositions}

\begin{definition}[Flow network]
\label{def:flownetwork}
A tuple $G=(V,E,f)$ is said to be a \emph{flow network} if $(V,E)$ is a graph with unique source $s$ and unique sink $t$, where for every edge $(u,v) \in E$ we have an associated positive integer \emph{flow value} $f_{uv}$, satisfying \emph{conservation of flow} for every $v \in V \setminus \{s,t\}$, namely:
\begin{equation} \label{eqn:conservation_of_flow}
\sum_{(u,v) \in E} f_{uv} = \sum_{(v,w) \in E} f_{vw}.
\end{equation}
\end{definition}

Next, we define when a set of generic walks in a flow network form a flow decomposition. By considering different types of walks, we will obtain the specific problem variants of this paper.

\begin{definition}[$k$-Flow Decomposition]

A \emph{$k$-flow decomposition} $(\walks,w)$ for a flow network $G=(V,E,f)$ is a set of $k$ walks $\walks = (W_1,\ldots,W_k)$ in $G$
and associated weights $w = (w_1,\ldots,w_k)$,
with each $w_i \in \mathbb{Z}^{+}$, such that for each edge $(u,v) \in E$ it holds that:
\begin{equation}
\label{eqn:flow_eq}
\sum_{i \in \{1,\dots,k\}} w_iW_i(u,v) = f_{uv}.
\end{equation}
The number $k$ of walks is also called the \emph{size} of the flow decomposition.

\label{def:flow-decomposition-walks}
\end{definition}

We will consider several types of walks, and obtain corresponding variants of the flow decomposition problem (see also \Cref{fig:FD}). We will also refer to the walks of a decomposition as its \emph{elements}. Next, we define the flow decomposition problems by asking for decompositions into \emph{at most} $k$ elements. 

\begin{figure}
	\centering\vspace{-2pt}
	\subfloat[A flow network]{\includegraphics[width=0.4\textwidth]{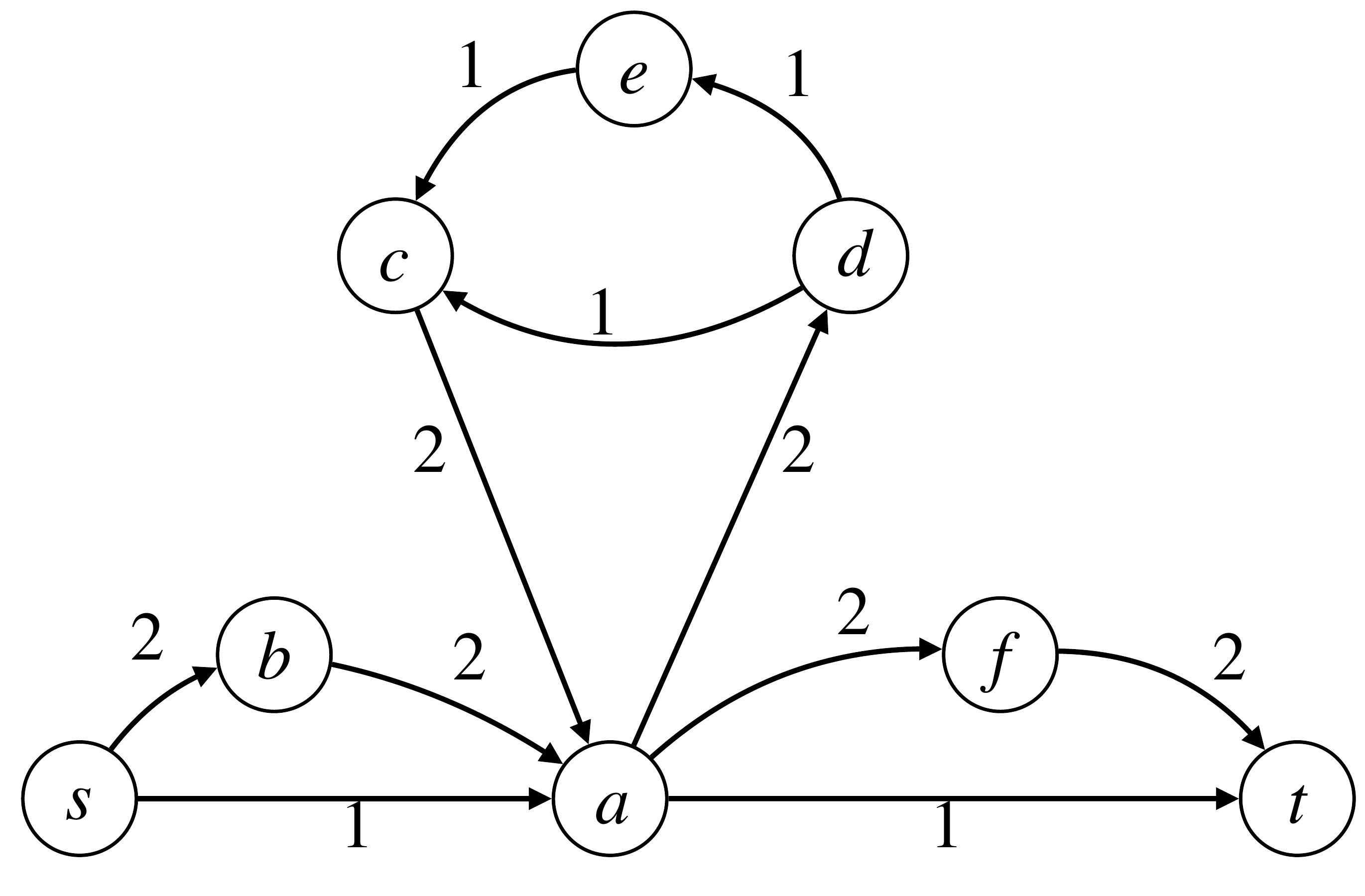}\label{fig:net}}
	\hspace{1cm}
	\subfloat[A 4-flow decomposition into $s$-$t$ paths and cycles of weights $(2,1,1,1) $]{\includegraphics[width=0.4\textwidth]{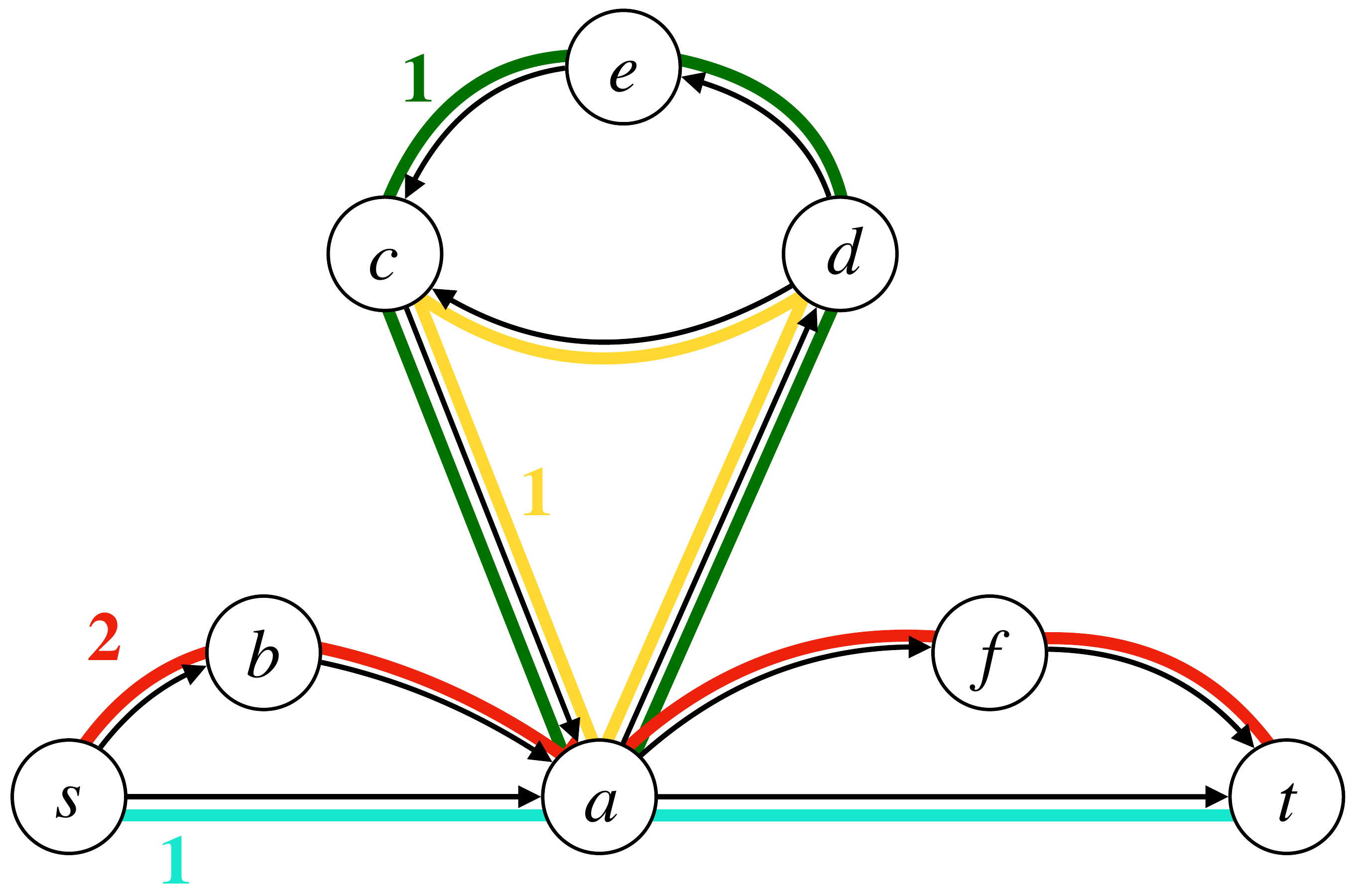}\label{fig:4-FD}}
	\\
	\subfloat[A 3-flow decomposition into $s$-$t$ trails of weights
	$(1,1,1)$]{\includegraphics[width=0.4\textwidth]{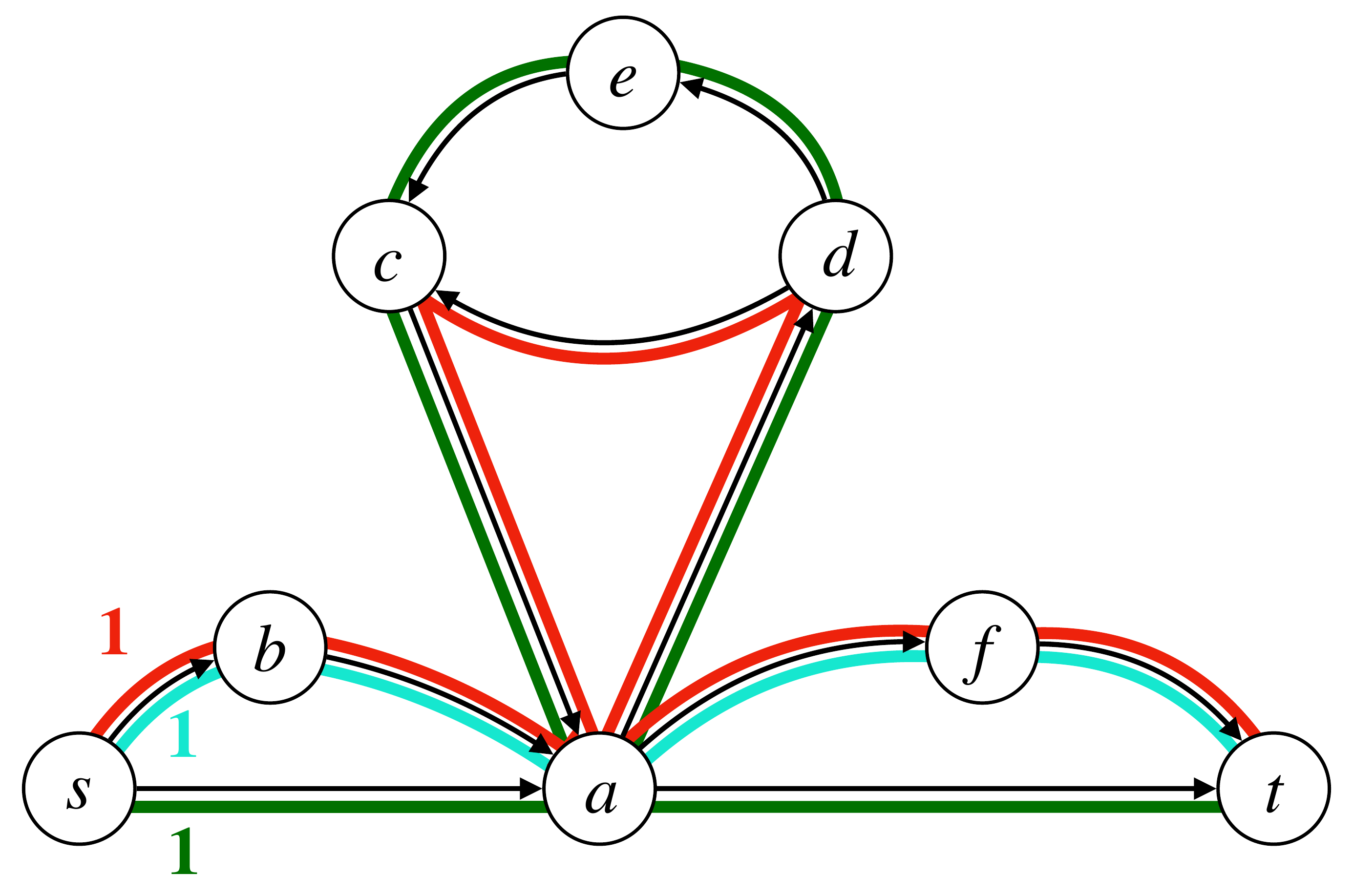}\label{fig:3-FD}} 
	\hspace{1cm}
	\subfloat[A 2-flow decomposition into $s$-$t$ walks of weights $(2,1)$.]{\includegraphics[width=0.4\textwidth]{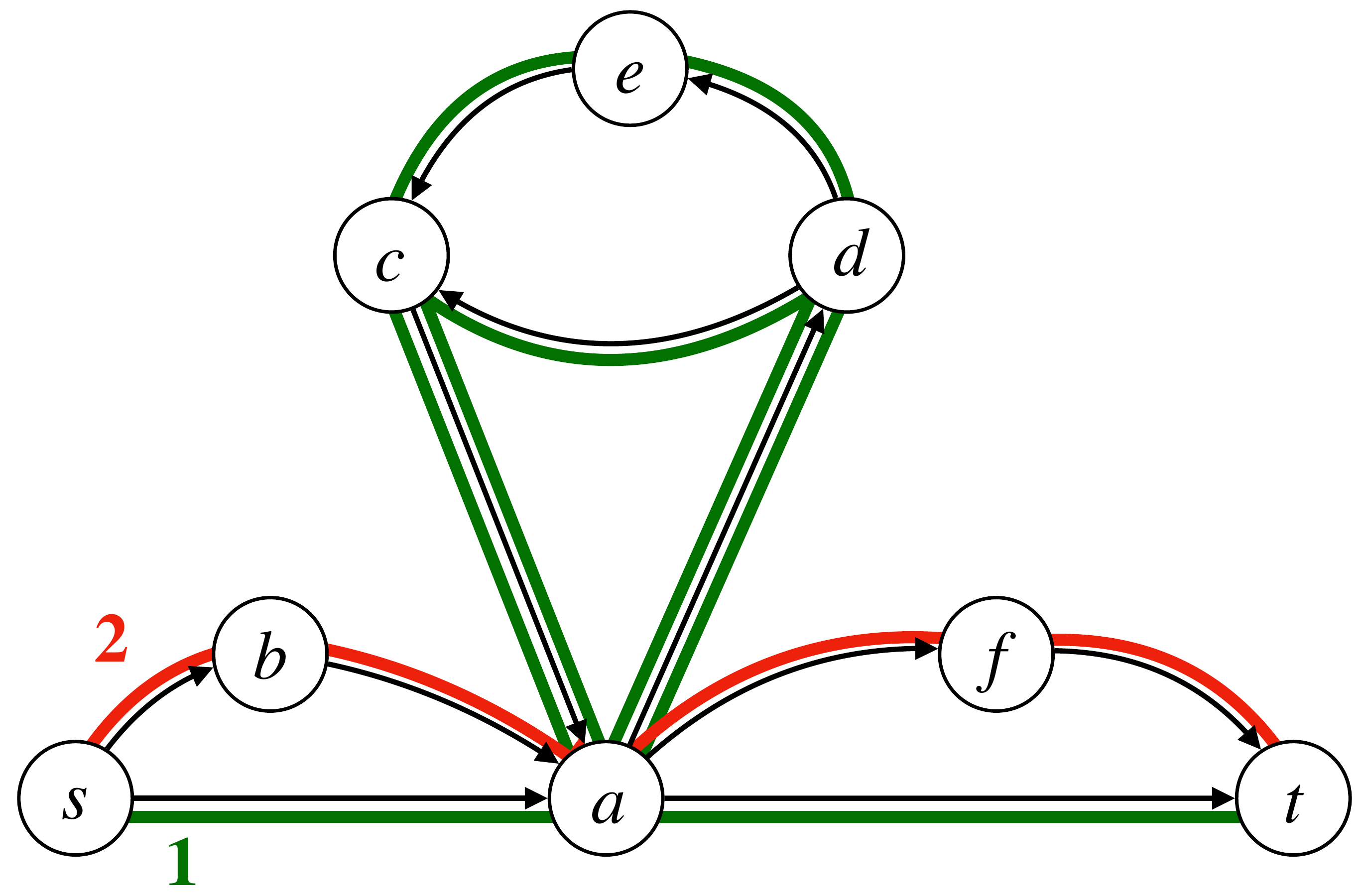}\label{fig:2-FD}}\\
	\caption{Example of a flow network (\ref{fig:net}) and of three different types of minimum flow decomposition into paths or cycles (\ref{fig:4-FD}), trails (\ref{fig:3-FD}) and walks (\ref{fig:2-FD}).\label{fig:FD}
	}
t\end{figure}

\begin{definition}[$k$-Flow Decomposition Problems]
\label{def:problem-variants}
Given a flow network $G=(V,E,f)$, consider the problem of finding a flow decomposition into at most $k$ walks and associated weights. If these walks are required to be:
\begin{itemize}
    \item $s$-$t$ paths or cycles, then we call the resulting problem \emph{$k$-Flow Decomposition into Paths or Cycles ($k$-FDPC)};
    \item $s$-$t$ trails, then we call the resulting problem \emph{$k$-Flow Decomposition into Trails ($k$-FDT)};
    \item $s$-$t$ walks, then we call the resulting problem \emph{$k$-Flow Decomposition into Walks ($k$-FDW)}.
\end{itemize}
\end{definition}

It is a standard result (see e.g.~\cite{ahuja1988network}) that any flow network can be decomposed into paths or cycles because one can consider any edge $(u,v)$, and by conservation of flow, extend this edge in both directions either into an $s$-$t$ path, or into a cycle, associate weight $1$ to it, and remove it from the flow network, obtaining another flow network. We can analogously deduce that any flow network can also be decomposed into a collection of $s$-$t$ walks. The simplest proof of this fact can be obtained by creating another graph by replacing every edge $(u,v)$ with $f(u,v)$ parallel edges and adding as many parallel edges from $t$ to $s$ as there is flow out-going from $s$. Since in this graph, every node has as many in-neighbours as out-neighbours (by construction and by conservation of flow), the resulting graph has an Eulerian walk with the first and last node equal to $s$. Every time this walk passes through one of the parallel edges from $t$ to $s$ we cut it by removing such edges. The resulting pieces are $s$-$t$ walks covering all edges, thus forming a flow decomposition of the original graph, where each walk is associated with weight 1.

However, not all flow networks admit flow decomposition into $s$-$t$ trails. For example, this is the case for the flow network made up of edges $(s,a),(a,d),(d,c),(c,a),(a,t)$, with flow values $f(s,a) = 1$, $f(a,d) = 2$, $f(d,c) = 2$, $f(c,a) = 2$ and $f(u,t) = 1$ (see also \Cref{fig:no_trails}). 

\begin{figure}
    \centering
    \includegraphics[scale=0.25]{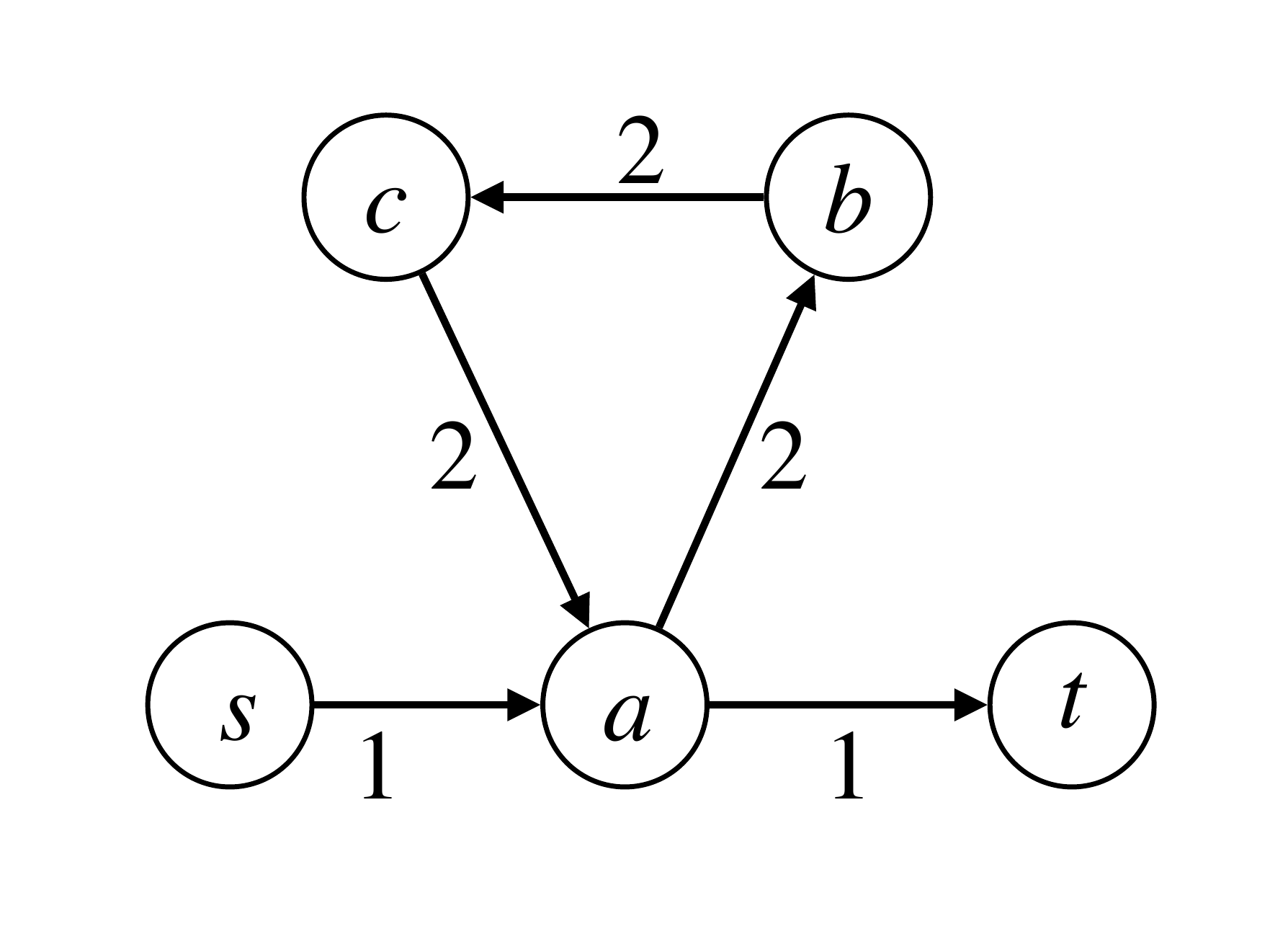}
    \caption{Example of a network that does not admit a decomposition into $s$-$t$ trails.
\label{fig:no_trails}
}
\end{figure}

In practice, we are interested in decompositions of minimum size (i.e., with the minimum number of elements, that is, walks, recall \Cref{def:flow-decomposition-walks}). As such, we can introduce the \emph{minimization} version of each problem, as follows.

\begin{definition}[Minimum Flow Decomposition (MFD)]
\label{def:min-flow-decomposition}
    For each of the problems $k$-FDPC, $k$-FDT, $k$-FDW, its \emph{minimization} version asks finding a flow decomposition of minimum size (i.e., minimizing $k$).
\end{definition}

All our ILP formulations will assume at most $k$ paths/cycles, trails, or walks. As such, the minimization versions of the problems can be solved via an iterative process, trying all possible $k$ in increasing order until reaching the minimum size $k^\ast$ of a decomposition (or $m$, which is an upper bound on the size of an MFD in all problem variants, as discussed above). When $k > m$, it means that there is no feasible solution. This iterative approach leads to $k^\ast$ runs of our ILPs.

 However, a more efficient approach, leading to only $O(\log k^\ast)$ runs of our ILPs, is based on the fact that if a decomposition of size at most $k$ exists, then trivially also one of size at most $k' > k$ exists as well.
 As such, we can first do a \emph{doubling} (or \emph{exponential}) approach~\citep{bentley1976almost} by running our ILPs with $k = 1, 2, 4, \dots$, until finding the smallest $k$ such that a decomposition of size at most $2k$ exists, but one of size $k$ does not exists. Then, it holds that $k^\ast \in (k,2k]$, and we can find $k^\ast$ with a binary search in this interval.

\subsection{ILP formulation of $s$-$t$ paths in DAGs and of flow decompositions}
\label{sec:ilp-fd-dags}

In this section we review the formulation from~\cite{taccari2016integer} of an $s$-$t$ path in a graph $G = (V,E)$, with $s,t \in V$, and then recall the use of this formulation in the ILP solution for MFD in DAGs~\citep{dias2022fast}.

For every edge $(u,v) \in E$, we can introduce a binary variable $x_{uv}$. The idea is to represent the $s$-$t$ path by the edges $(u,v)$ having $x_{uv} = 1$. If $G$ is a DAG, it suffices to impose the following constraint, stating that the $s$-$t$ path starts with a single edge out-going from $s$, ends with a single edge in-coming to $t$, and at every other node $v$, either it does not pass through the node, or if it reaches $v$, it also exits $v$. Note that since $G$ is a DAG, no walk can use a node $v$ multiple times.

\begin{equation}
\label{eqn:taccari-fc}
    \sum_{(u,v) \in E} x_{uv} - \sum_{(v,u) \in E} x_{vu} = 
    \begin{cases}
    0, & \text{if $v \in V \setminus \{s,t\}$}, \\
    1, & \text{if $v = t$}, \\
    -1, & \text{if $v = s$}.
    \end{cases}
\end{equation}
We will refer to the above constraint as the ``flow conservation'' constraint (not to be mistaken with the conservation of the input flow). In~\cite{dias2022fast}, a flow decomposition in DAGs into exactly $k$ $s$-$t$ paths is expressed by adding $k$ copies of the above variables, namely $x_{uvi}$, for $i \in \{1,\dots,k\}$, and imposing the flow conservation constraint \eqref{eqn:taccari-fc} for all of them. The flow decomposition constraint \eqref{eqn:flow_eq} can be stated as
\begin{align}
\label{eqn:ilp_flow_eq}
& \sum_{i \in \{1,\dots,k\}} x_{uvi}w_i = f_{uv}, && \forall (u,v) \in E.
\end{align}

We refer to this ILP formulation for MFD in DAGs as the \emph{standard formulation}; our formulations will build on this, mainly by possibly changing \cref{eqn:taccari-fc} and adding further constraints to the $x_{uvi}$ variables to model all types of walks considered in this paper.

\begin{remark}
\label{rem:lin}
The constraint in Eq.~\eqref{eqn:ilp_flow_eq} can be linearized by introducing an integer variable $\pi_{uvi}$ representing the product $x_{uvi}w_i$, together with the following constraints, where $M$ is any constant strictly greater than the right-hand side of the first inequality (\cite{soltysik2010two}) in Eq.~\eqref{eq:flow_lin_b}:

\begin{subequations}
\begin{align}
& && f_{uv} = \sum_{i \in \{1,\dots,k\}} \pi_{uvi}, && \forall (u,v) \in E, \label{eq:flow_lin_a}\\
& && \pi_{uvi} \leq M x_{uvi}, && \forall (u,v) \in E, \forall i \in \{1,\dots,k\},\label{eq:flow_lin_b}\\
& && \pi_{uvi} \leq w_i, && \forall (u,v) \in E, \forall i \in \{1,\dots,k\},\label{eq:flow_lin_c}\\
& && \pi_{uvi} \geq w_i - (1-x_{uvi})M, && \forall (u,v) \in E, \forall  i \in \{1,\dots,k\}.\label{eq:flow_lin_d}
\end{align}
\label{int_lin2}
\end{subequations}

Indeed, in inequality \eqref{eq:flow_lin_b}, for $x_{uvi}=1$, the constraint 
should be bounded by the flow value $f_{uv}$.
If $x_{uvi}=0$, then the constraint is not binding, therefore the left side can be any value that is smaller or equal to zero. For constraint \eqref{eq:flow_lin_d}, when $x_{uvi}=1$, the value of $M$ is irrelevant; when $x_{uvi}=0$, the 
constraint is not binding as long as $M$ is a positive number. Therefore, the best value of $M$ is $M = f_{uv}$. 
\end{remark}

\section{ILP Formulations for MFD in Graphs with Cycles}
\label{ilp}

\subsection{$k$-Flow Decomposition into Paths or Cycles}
\label{pc}

In this section we tackle the $k$-FDPC problem. We start as in the standard formulation, by using the same binary variables $x_{uvi}$ for every edge $(u,v) \in E$ and every $i \in \{1,\dots,k\}$. Since paths and cycles can visit any node at most once, we adapt the flow conservation constraint \eqref{eqn:taccari-fc} into the following constraints:
\begin{subequations}
\label{eq:path_flow}
\begin{align}
& && \sum_{(u,v) \in E} x_{uvi} \leq 1 && \forall u \in V, \forall i \in \{1, \ldots, k\}, \label{eq:pc:s_cons}\\
& && \sum_{(u,v) \in E} x_{uvi} - \sum_{(v,w) \in E} x_{vwi} = 0 &&  \forall i \in \{1, \ldots, k\}, \forall v \in V \setminus \{s, t\}.\label{eq:pc:fc}
\end{align}
\end{subequations}

However, if $G$ is not a DAG, the binary variables $x_{uvi}$ could for example induce an $s$-$t$ path and one or more cycles (not even reachable from the $s$-$t$ path). 

To overcome this, we start from a formulation of an $s$-$t$ path in general graphs, presented in~\cite{taccari2016integer} as the \emph{sequential formulation}, and attributed to~\cite{miller1960integer}. This formulation introduces a positive integer variable $d_v$ for every node $v \in V$, which will stand for the position of $v$ in the path. The following constraint is introduced for every edge:
\begin{align}
\label{eq:taccari-sequential}
    d_v \geq d_u + 1 + (n-1)(x_{uv} - 1) && \forall (u,v) \in E.
\end{align}

This assigns a partial order to the nodes, such that for any edge $(u,v)$ with $x_{uv} = 1$, we have $d_u < d_v$. We can also think of the $d_v$ variables as a \emph{pseudo-distance} function, where the true distance from $s$ to $t$ is obtained by minimizing $d_t$. However, constraint \eqref{eq:taccari-sequential} as is forbids cycles altogether, which is not the goal in the $k$-FDPC problem.

\begin{figure}
    \centering
    \includegraphics[width=1.0\textwidth]{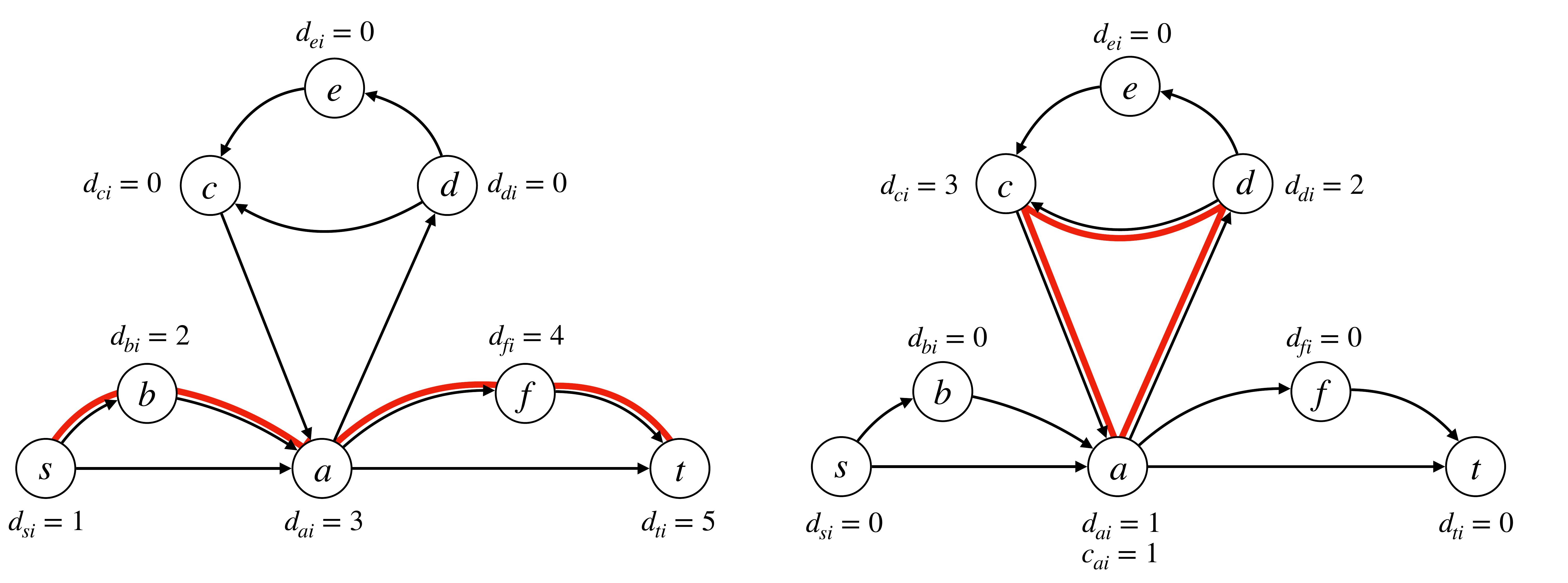}
    \caption{
    Example of the edge variables $x_{uvi}$ (where we draw edge $(u,v)$ in red if $x_{uvi} = 1$, and in black otherwise), satisfying \cref{eq:pc:fc,eq:pc:s_cons}. On the left, the assignments of the variables $d_{vi}$ 
    shown in the figure (and $c_{vi} = 0$ for all $v \in V$) satisfy constraint \eqref{eq:sequential+cycle}. On the right, the assignments of the variables $d_{vi}$ shown in the figure (and $c_{vi} = 1$ for all $v \in V \setminus \{a\}$, and $c_{ai} = 1$) satisfy constraint \eqref{eq:sequential+cycle}. Note that the only edge $(u,v)$ such that $d_{vi} < d_{ui}$ is $(c,a)$. However, constraint \eqref{eq:sequential+cycle} holds for $(c,a)$ because $c_{ai} = 1$.
    }
    \label{fig:fdpc}
\end{figure}

In order to handle both paths and cycles, we modify the sequential formulation constraint \eqref{eq:taccari-sequential} as follows. For every $v \in V \setminus \{s,t\}$, we introduce a binary variable $c_{vi}$ to denote whether node $v$ is a ``start of cycle'' node in the $i$-th element of the flow decomposition. We also introduce positive integer variables $d_{vi}$, and change \eqref{eq:taccari-sequential} into the following constraint:
\begin{align}
\label{eq:sequential+cycle}
    d_{vi} \geq d_{ui} + 1 + (n-1)(x_{uvi} - 1 - c_{vi}) && \forall (u,v) \in E, \forall i \in \{1,\ldots,k\}.
\end{align}
Namely, this constraint imposes that either a node is the starting point of a cycle or its order in the path is always increasing. Finally, we impose:
\begin{align}
\label{eq:paths+cycles-together}
    \sum_{v\in V} x_{svi} + \sum_{v\in V} c_{vi} \leq 1 && \forall i \in \{1,\dots,k\}.
\end{align}
That is, for every element $i$ of the decomposition, either there is some edge $(s,v)$ with $x_{svi} = 1$ (in which case we have just an $s$-$t$ path, and no cycle), or $c_{vi} = 1$ for some $v$ (or none of these hold) (see \Cref{fig:fdpc}). Because of constraint \eqref{eq:pc:fc}, if $\sum_{v\in V} x_{svi} = 1$, then the binary $x_{uvi}$ variables induce just an $s$-$t$ path and no other cycles. If $c_{vi} = 1$, this means that we have no $s$-$t$ path and just one cycle (or none at all, if all $x_{uvi}$ variables are zero). Finally, observe that adding constraint \eqref{eq:paths+cycles-together} we obtain that  constraint \eqref{eq:pc:s_cons} is always satisfied, since otherwise, we have either at least two $s$-$t$ paths passing through $u$ in some element $i$ of the decomposition (and thus $\sum_{v\in V} x_{svi} \geq 2$), or two cycles passing through $u$ in some element $i$ of the decomposition (and thus $\sum_{v\in V} c_{vi} \geq 2$). %\lucycom{If constraint \eqref{eq:pc:s_cons} is always satisfied, why include it in the full model?}

The complete ILP formulation for this problem variant, which we give as Model~\ref{mod:pc} in \ref{sec:k-fdpc}, has $O((|V|+|E|)k)$ variables and constraints.

\begin{remark}
The formulation for $k$-FDPC presented here is capable of decomposing the flow into up to $k$ paths or cycles. However, one can easily obtain a formulation where there are \emph{exactly} $k$ elements in the decomposition, as follows. The constraint in \eqref{eq:paths+cycles-together} allows the number of total cycles and paths to be up to $k$, by allowing either terms of the left-hand side to be at most 1. Therefore, this constraint needs to be changed to:
\begin{align}
\label{eq:paths+cycles-together2}
    \sum_{v\in V} x_{svi} + \sum_{v\in V} c_{vi} = 1 && \forall i \in \{1,\dots,k\}.
\end{align}
\end{remark}

\subsection{$k$-Flow Decomposition into Trails via Constraint Generation}
\label{pt}

In this section we tackle the next problem, $k$-FDT, namely finding a $k$-flow decomposition into at most $k$ $s$-$t$ trails. Recall from the previous sections that the binary $x$ variables induce a flow. The only case when such flow does not correspond to an $s$-$t$ trail is when it has a strongly connected component that is not reachable from $s$, or equivalently, that does not reach $t$.

We start with a simple lemma stating that, due to flow conservation, this global reachability property can be reduced to checking whether any strongly connected component $C$ induced by the edges of the trail has an edge out-going from $C$ (i.e., an edge $(u,v)$ with $x_{uvi} = 1$, such that $u$ appears in $C$, but $(u,v)$ does not belong to $C$). For a strongly connected component $C$ of $G$, we denote by $E(C)$ the set of edges of $C$ and by $\delta^+(C)$ the set of edges $(u,v)$ such that $u$ belongs to $C$.

\begin{lemma}
\label{lem:trails-outgoing-edge}
Let $G = (V,E)$ be an arbitrary graph with source and sink nodes $s$ and $t$. Let $W$ be a set of edges of $G$, where for every edge $(u,v) \in E$ we define $W(u,v) := 1$, if $(u,v) \in W$, and $W(u,v) := 0$ otherwise. It holds that the edges of $W$ can be ordered to obtain a $s$-$t$ trail passing through each edge of $W$ if and only if the following conditions hold:
\begin{enumerate} 
    \item For every $v \in V$,
        $\displaystyle\sum_{(u,v) \in E} W(u,v) - \sum_{(v,u) \in E} W(v,u) = 
        \begin{cases}
        0, & \text{if $v \in V \setminus \{s,t\}$}, \\
        1, & \text{if $v = t$}, \\
        -1, & \text{if $v = s$}.
        \end{cases}$
    \item For any strongly connected component $C$ of the subgraph of $G$ induced by the edges in $W$ and different from $\{t\}$, $\delta^+(C) \setminus E(C) \cap W$ is non-empty.    
\end{enumerate}
\end{lemma}

\begin{proof}
For the forward implication, suppose the edges of $W$ can be ordered to obtain an $s$-$t$ trail $T$ that satisfies condition 1. To see that $W$ satisfies also condition 2., let $C$ be a strongly connected component of the subgraph of $G$ induced by the edges of $W$, different from $\{t\}$. First, note that node $t$, being a sink, does not belong to $C$. 

If $C$ contains a single node, say $v \neq t$, then the edge of the trail out-going from $v$ belong to $\delta^+(C) \setminus E(C) \cap W$. Otherwise, let $(u,v)$ be the edge of $C$ that is last in the order given by the trail $T$. Say that $(v,w)$ is the edge following $(u,v)$ in $T$ (which exists because $v \neq t$). We have that $(v,w) \in \delta^+(C)$ (since $v$ belong to $C$), $(v,w) \notin E(C)$ (since $(u,v)$ is the last edge of $C$ in the order given by $T$, and $(u,v) \in W$, since $(u,v)$ is an edge of the trail. Thus, the set $\delta^+(C) \setminus E(C) \cap W$ is non-empty, and $W$ satisfies Condition 2.

For the reverse implication, suppose $W$ satisfies the two conditions. Since the values $W(u,v)$ induce a flow in $G$, it remains to prove that this flow is decomposable into a \emph{single} $s$-$t$ trail. Since $\sum_{(u,t) \in E} W(u,t) - \sum_{(t,u) \in E} W(s,u) = 1$, the only way this does not hold is when $W$ also contains a set of edges forming a strongly connected component, and no node in this component reaches $t$. This means that all edges out-going from the nodes of $C$ are to other nodes of $C$ (by flow conservation, and by the fact that $t$ is not reachable from $C$). However, this contradicts the second condition. See the example in \Cref{fig:ex_trail}. \qed
\end{proof}

\begin{figure}
\centering
\begin{subfigure}[c]{0.45\textwidth}
\includegraphics[width=\textwidth]{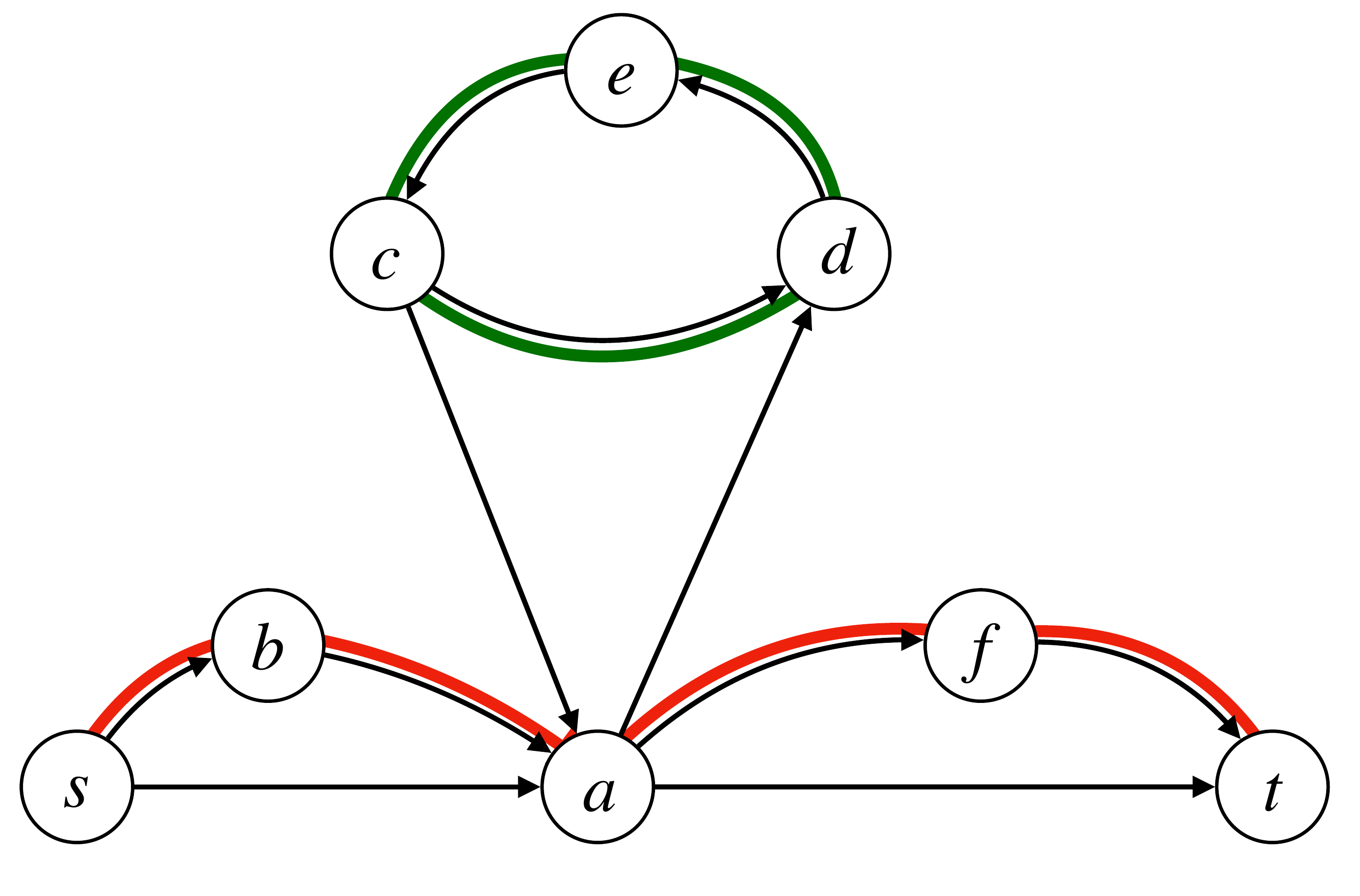}
\caption{$W_1$}
\label{iter_trailA}
\end{subfigure}
\begin{subfigure}[c]{0.45\textwidth}
\includegraphics[width=\textwidth]{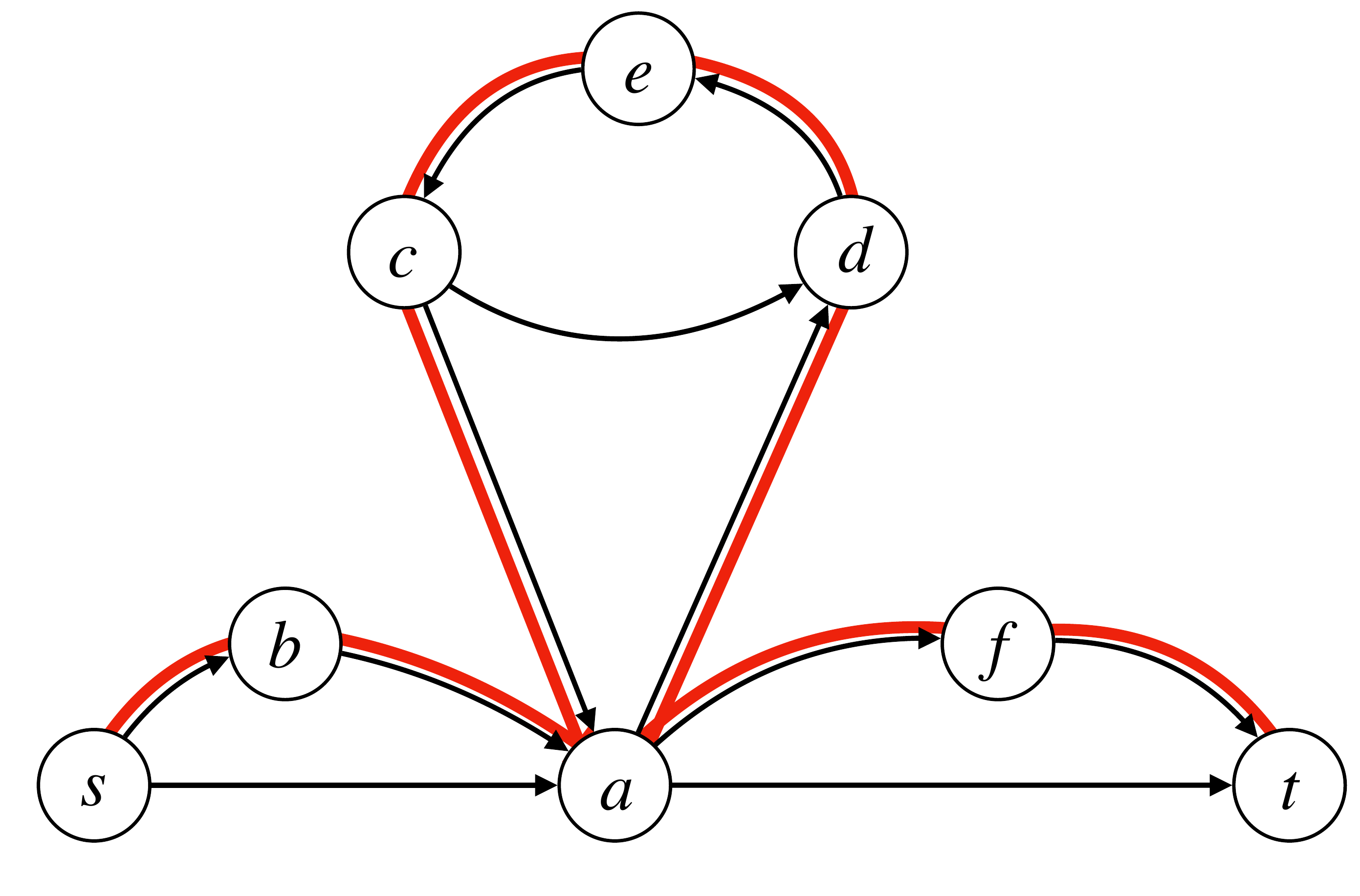}
\caption{$W_2$}
\label{iter_trailB}
\end{subfigure}
\caption{The graph from~\Cref{fig:FD}, with the edge $(d,c)$ reversed. This creates the cycle $(c,d,e,c)$ in the graph, which can lead to a set $W_1$ of edges (on the left, with edges in both red and green), satisfying condition 1. of \Cref{lem:trails-outgoing-edge}, but with a strongly connected component (in green) violating condition 2. of \Cref{lem:trails-outgoing-edge}. Instead, the set $W_2$ (on the right, with edges in red) satisfies both conditions of \Cref{lem:trails-outgoing-edge}.\label{fig:ex_trail}}
\end{figure}

As in the standard formulation, for every edge $(u,v) \in E$ and for every $i \in \{1,\dots,k\}$, we introduce a binary variable $x_{uvi}$ that will equal 1 if and only if the $s$-$t$ trail $W_i$ passes through $(u,v)$ (since trails can also visit edges at most once). The first condition of \Cref{lem:trails-outgoing-edge} is the same as the flow conservation condition in \Cref{eqn:taccari-fc}. However, as in the previous formulation, since we are modeling up to $k$ trails, we impose it in its weaker form below. (At the end of this section (\Cref{rem:trails-exactly-k}) we explain how to model \emph{exactly} $k$ trails.)

\begin{align}
& && \sum_{(s,v) \in E} x_{svi} \leq 1 && \forall i \in \{1, \ldots, k\},\label{eq:flow_trails_1}\\
& && \sum_{(u,v) \in E} x_{uvi} - \sum_{(v,w) \in E} x_{vwi} = 0 &&  \forall i \in \{1, \ldots, k\}, \forall v \in V \setminus \{s, t\}.\label{eqn:flow-conservation-trail}
\end{align}

To model the second condition of \Cref{lem:walks-reachability}, a first inefficient method would be to first exhaustively enumerate possible strongly connected subgraphs of $G$; their number is finite, since $G$ has a finite number of nodes. For each such strongly connected subgraph $C$, add constraints to prevent it from not reaching $t$, via the second condition of \Cref{lem:trails-outgoing-edge}. The following constraint models the second condition of \Cref{lem:walks-reachability} for a given strongly connected component $C$, where $|C|$ denotes the number of edges of $C$ (see \Cref{remark:linearize-trails-elim} on how to linearize this constraint):

\begin{align}
&\sum_{(u,v) \in E(C)} x_{uvi} = |C|  \rightarrow \sum_{(u,v) \in \delta^+(C) \setminus E(C)} x_{uvi} \geq 1 && \forall i \in \{1,\dots,k\}. \label{eq:trails-elim}
\end{align}

In order to avoid exhaustively enumerating all possible strongly connected subgraphs $C$ of $G$, and imposing constraint \eqref{eq:trails-elim} for each such $C$, we can consider the following iterative procedure, in which  constraint \eqref{eq:trails-elim} is not imposed at the beginning, but it is iteratively imposed for any strongly connected component induced by the variables $x_{uvi}$ not satisfying the constraint.

More precisely, we can consider a parameterized version of our ILP denoted as $k$-FDC($\mathcal{C}$) made up of the constraints in Eq.~\eqref{eqn:ilp_flow_eq},\eqref{eq:flow_trails_1},\eqref{eqn:flow-conservation-trail},\eqref{eq:trails-elim}, where $\mathcal{C}$ is a set of strongly connected components for which constraint \eqref{eq:trails-elim} is imposed. Initially, $\mathcal{C}$ is empty. For any strongly connected component $C$ induced by the variables $x_{uvi}$ not satisfying \eqref{eq:trails-elim}, we add $C$ to $\mathcal{C}$, and run the ILP again. This procedure eventually stops, because the number of strongly connected subgraphs of $G$ is finite, as mentioned above. We summarize this procedure also as \cref{algo:trail} (see also \Cref{fig:Trails_Iter}). This iterative addition of constraints to an initial formulation is known in the literature as a \emph{constraint generation} (or \emph{cut generation}), and it is based on the cycle elimination technique first proposed by \cite{dantzig1954solution} for the Traveling Salesman Problem.

The full ILP formulation for $k$-FDT($\mathcal{C}$), which we give as Model~\ref{mod:pt} in \ref{sec:k-fdt}, has $O((|E|+|\mathcal{C}|)k)$ variables and $O((|V|+|V|+|\mathcal{C}|)k)$ constraints.

\begin{remark}
\label{remark:linearize-trails-elim}
In order to linearize \eqref{eq:trails-elim}, we first rewrite it as the following disjunction:
\begin{align}
& \sum_{(u,v) \in E(C)} x_{uvi} \leq |C| - 1  \vee \sum_{(u,v) \in \delta^+(C) \setminus E(C)} x_{uvi} \geq 1 && \forall C \in \mathcal{C}, \forall i \in \{1,\dots,k\}.
\label{eq:trails-disjunction}
\end{align}

To linearize this disjunction, we introduce the binary variables $\beta_{Ci}$, for the strongly connected component $C$, and each $i \in \{1,\dots,k\}$. Following standard techniques (as applied in \cref{rem:lin}), let $M$ be any constant strictly greater than the right-hand side of the first inequality in \eqref{eq:trails-disjunction}. We can rewrite \eqref{eq:trails-disjunction} as the following three constraints:

\begin{subequations}
\begin{align}
&\sum_{(u,v) \in E(C)} x_{uvi} \geq |C| - M(1-\beta_{Ci})  && \forall C \in \mathcal{C}, \forall i \in \{1,\dots,k\},\label{eq:m1}\\
& \sum_{(u,v) \in E(C)} x_{uvi} - |C| + 1 - M\beta_{Ci} \leq 0 && \forall C \in \mathcal{C}, \forall i \in \{1,\dots,k\},\label{eq:m2}\\ 
& \sum_{(u,v) \in \delta^+(C) \setminus E(C)} x_{uvi} \geq \beta_{Ci} && \forall C \in \mathcal{C}, \forall i \in \{1,\dots,k\},\\ 
& \beta_{Ci} \in \{0,1\}  && \forall C \in \mathcal{C}, \forall i \in \{1,\dots,k\}. \label{eq:m4}
\end{align}
\label{eq:sub_trail}
\end{subequations}

For $M$, a suitable value would be $M = |C|$. In \cref{eq:m1} and \cref{eq:m2}, the left side of each constraint should be not bounded when $\beta_{Ci}=1$ and $\beta_{Ci}=0$, respectively. Hence, the smallest number that guarantee this condition is $M = |C|$.
\end{remark}

\begin{remark}
The check in \Cref{algo:trail} for strongly connected components not satisfying constraint \eqref{eq:trails-elim} can be implemented efficiently (in linear time), as follows. Consider the graph $G'$ where we add all the edges $(u,v)$ of $G$ where the $x_{uvi}$ variables are set to 1, plus the additional edge $(t,s)$. Compute the strongly connected components of $G'$, doable in linear time~\cite{tarjan1972depth}. Using arguments as in the proof of \Cref{lem:trails-outgoing-edge}, one can easily see that $G'$ has only one strongly connected component if and only if the check in \Cref{algo:trail} is false (i.e., the $x_{uvi}$ variables do not induce any strongly connected component $C$ violating~\eqref{eq:trails-elim}).
% The check in \Cref{algo:trail} for strongly connected components not satisfying constraint \eqref{eq:trails-elim} can be optimized as follows. \textcolor{blue}{If $t$ is part of a strongly connected component $C$} \lucycom{what does it mean for a SCC to reach $t$? Isn't $t$ always its own SCC, since it has no edges outgoing?}, then the forward implication of the proof in \Cref{lem:trails-outgoing-edge} 
% guarantees that $\delta^+(C) \setminus E(C) \cap W$ is non-empty, and thus it satisfies constraint~\eqref{eq:trails-elim}. To find the strongly connected components \emph{not} reaching $t$, one can consider the graph induced by the $x_{uvi}$ variables set to 1, where we also add an edge from $t$ to $s$. All the strongly connected components not satisfying \eqref{eq:trails-elim} are the strongly connected components of this graph not containing $t$.
\end{remark}

\begin{algorithm}
\allowdisplaybreaks
\SetKw{kwContinue}{continue}
\SetKw{kwReturn}{return}
\SetKw{kwInfeasible}{infeasible}
\SetKw{kwFeasible}{feasible}
\SetKwRepeat{Do}{repeat}{until}
\KwIn{Flow network $G=(V,E,f)$, and integer $k$}
\KwOut{A decomposition into at most $k$ trails with associated weights}
$\mathcal{C} \leftarrow \emptyset$\\
\Do{True}{
Solve $k$-FDT($\mathcal{C}$)\\
\If{\kwInfeasible}{ 
    \kwReturn \kwInfeasible
}  
 
\eIf{variables $x_{uvi}$ induce a strongly connected component $C$ not satisfying constraint~\eqref{eq:trails-elim}}{ 
        $\mathcal{C} \leftarrow \mathcal{C} \cup\{C\}$
}
{
    \kwReturn the trails and their associated weights induced from the $x$ and $w$ variables
}
}
  
\caption{Algorithm for $k$-Flow Decomposition into Trails via Constraint Generation} 
\label{algo:trail}
\end{algorithm}

\begin{figure}
\centering
\begin{subfigure}[c]{0.32\textwidth}
\includegraphics[width=\textwidth]{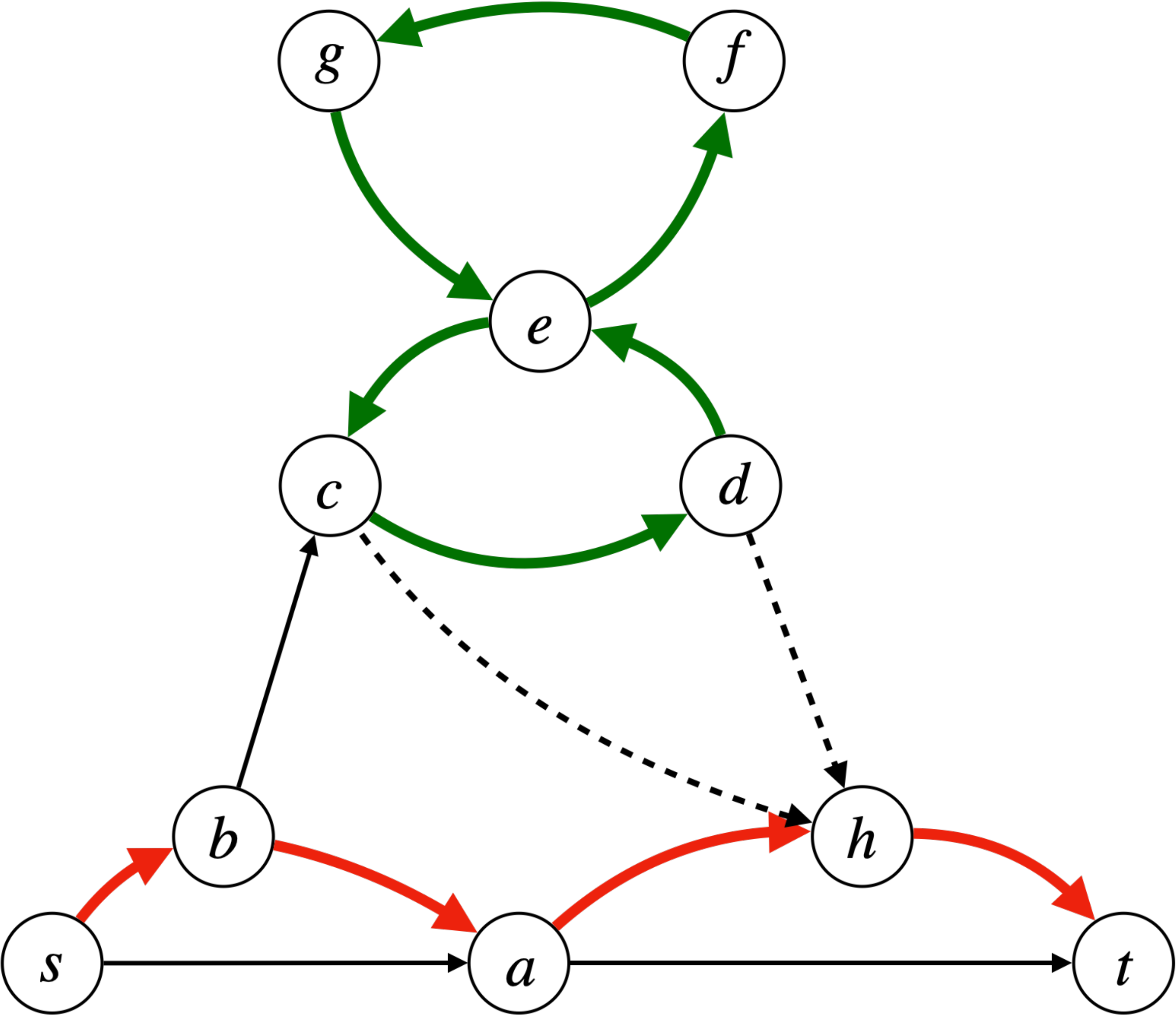}
\caption{First iteration}
\label{iter1}
\end{subfigure}
\hfill
\begin{subfigure}[c]{0.32\textwidth}
\includegraphics[width=\textwidth]{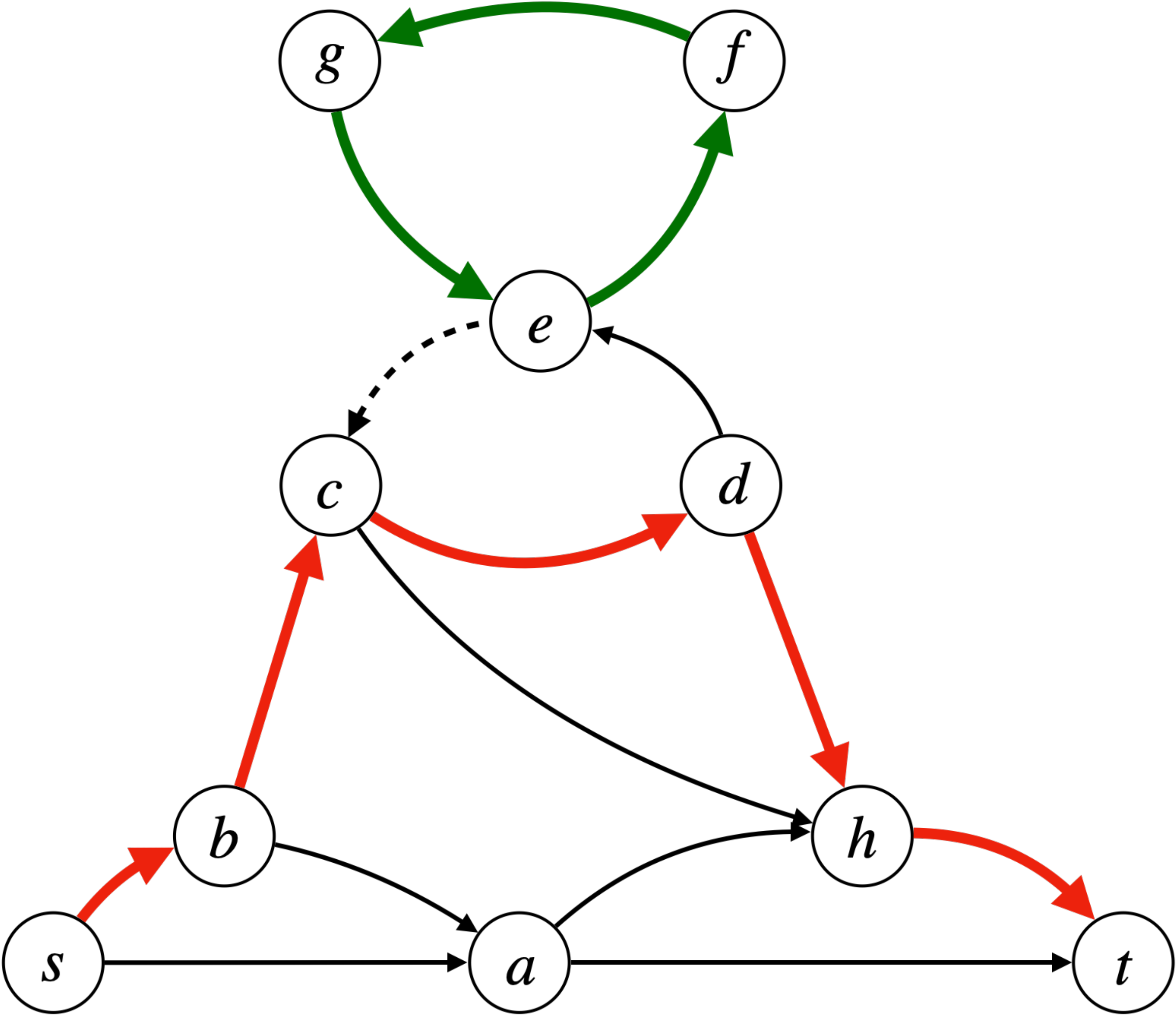}
\caption{Second iteration}
\label{iter2}
\end{subfigure}
\hfill
\begin{subfigure}[c]{0.32\textwidth}
\includegraphics[width=\textwidth]{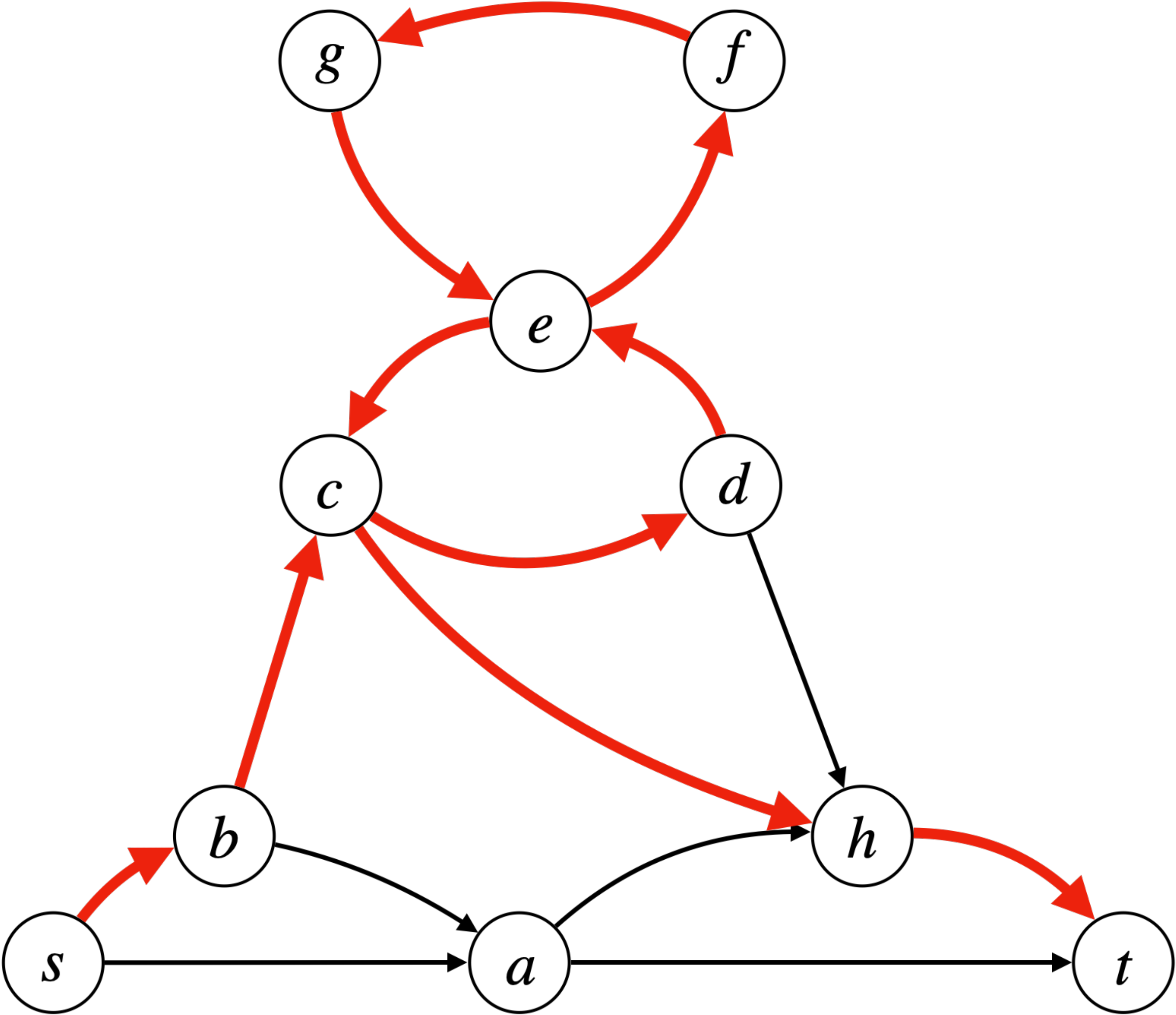}
\caption{Third iteration}
\label{iter3}
\end{subfigure}
\caption{
Illustration of the iterative process described in Algorithm \ref{algo:trail}. 
The edges in red form an $s$-$t$ trail while the edges in green from a strongly connected component that does not reach $t$ via edges of the trail (i.e. not satisfying constraint~\eqref{eq:trails-elim}). We assume that in all three iterations, the flow conservation and flow decomposition constraints are valid. In iteration one (Figure \ref{iter1}), there is a strongly connected component (in green). Based on the constraints imposed in Eq.~\eqref{eq:trails-elim}, at least one the outgoing edges from the strong connected component (dashed edges) needs to be selected. In Figure \ref{iter2}, after imposing those constraints, the original strongly connected component does not exist anymore. However, a new one is present in the solution and the same set of constraints need to imposed also for it. In Figure \ref{iter3}, the solution corresponds to an $s$-$t$ trail, where no isolated strongly connected component is present.}
\label{fig:Trails_Iter}
\end{figure}

\begin{remark}
\label{rem:trails-exactly-k}
In order to allow the formulation presented in this section to obtain a $k$-flow decomposition into \emph{exactly} $k$ $s$-$t$ trails, we need to change the flow conservation constraints in \Cref{eqn:flow-conservation-trail} and \Cref{eq:flow_trails_1} into its basic form as in \Cref{eqn:taccari-fc}:
\begin{equation}
\label{eqn:flow-conservation-trail2}
    \sum_{(u,v) \in E} x_{uvi} - \sum_{(v,u) \in E} x_{vui} = 
    \begin{cases}
    0, & \text{if $v \in V \setminus \{s,t\}$}, \forall i \in \{1,\dots,k\} \\
    1, & \text{if $v = t$}, \forall i \in \{1,\dots,k\}\\
    -1, & \text{if $v = s$}, \forall i \in \{1,\dots,k\}.
    \end{cases}
\end{equation}
\end{remark}

\subsection{$k$-Flow Decomposition into Trails and Walks}
\label{pw}

In this section we give a formulation that works for both problems $k$-FDT and $k$-FDW. Our formulations will be based on the following characterization of an $s$-$t$ walk, stating that all its nodes must be reachable from $s$, using only edges of the $s$-$t$ walk.

\begin{lemma}
\label{lem:walks-reachability}
Let $G = (V,E)$ be an arbitrary graph with source and sink nodes $s$ and $t$. Let $W$ be a multiset of edges of $G$,  where for every edge $(u,v) \in E$ we denote by $W(u,v)$ the number of times $(u,v)$ appears in the multiset $W$. It holds that the edges of $W$ can be ordered to obtain a $s$-$t$ walk passing $W(u,v)$ times through each edge $(u,v)$ if and only if the following conditions hold:
\begin{enumerate} 
    \item For every $v \in V$,
        $\displaystyle\sum_{(u,v) \in E} W(u,v) - \sum_{(v,u) \in E} W(v,u) = 
        \begin{cases}
        0, & \text{if $v \in V \setminus \{s,t\}$}, \\
        1, & \text{if $v = t$}, \\
        -1, & \text{if $v = s$}.
        \end{cases}$
    \item For every node $v \in V$ appearing in some edge of $W$, there is an $s$-$v$ path using only edges in $W$ (i.e.~$v$ is reachable from $s$ using edges in $W$).
\end{enumerate}
\end{lemma}

\begin{proof}
The forward implication is immediate: any $s$-$t$ walk satisfies the first condition of the lemma, and every node $v$ appearing in some edge of $W$ is reachable from $s$ via an $s$-$v$ path obtained by removing all cycles of the walk between $s$ and $v$.

For the reverse implication, we argue as in the proof of \Cref{lem:trails-outgoing-edge}. Suppose $W$ satisfies the two conditions. Since the values $W(u,v)$ induce a flow in $G$, it remains to prove that this flow is decomposable into a \emph{single} $s$-$t$ walk. Since $\sum_{(u,s) \in E} W(u,s) - \sum_{(s,u) \in E} W(s,u) = -1$, the only way this does not hold is when $W$ also contains a set of edges forming a strongly connected component, and no node in this component is reachable from $s$. However, this contradicts the second condition. See \Cref{fig:reachability-example}. \qed
\end{proof}

For every edge $(u,v) \in E$ and for every $i \in \{1,\dots,k\}$, we introduce a non-negative integer variable $x_{uvi}$ now representing the number of times the $s$-$t$ walk $W_i$ passes through $(u,v)$. Since walks can use edges arbitrarily many times, we impose $x_{uvi} \in \mathbb{Z}^+ \cup 0$ %\lucycom{union the set containing zero? or is this fine?}. 
To obtain the formulation for trails, it will suffice to impose $x_{uvi} \in \{0,1\}$, because a trail can visit every edge at most once.

As in the trails case from the previous section, we impose the flow conservation condition of \cref{lem:walks-reachability} in its weaker form (as in \Cref{eq:flow_trails_1} and \Cref{eqn:flow-conservation-trail}): 
\begin{subequations}
\label{eqn:flow-conservation-walk}
\begin{align}
& && \sum_{(s,v) \in E} x_{svi} \leq 1 && \forall i \in \{1, \ldots, k\},\label{eq:flow_walks_1}\\
& && \sum_{(u,v) \in E} x_{uvi} - \sum_{(v,w) \in E} x_{vwi} = 0 &&  \forall i \in \{1, \ldots, k\}, \forall v \in V \setminus \{s, t\}.
\end{align}
\end{subequations}

The second condition of \Cref{lem:walks-reachability} can be encoded as follows. We say that an edge $(u,v) \in E$ is \emph{selected} by $W_i$ if $x_{uvi} \geq 1$, and that a node $v$ is \emph{selected} by $W_i$ if $v$ has a selected edge in-coming to it (i.e. $x_{uvi} \geq 1$, for some $(u,v) \in E$).

For every $v \in V$ and every $i \in \{1,\dots,k\}$, we introduce a non-negative integer variable $d_{vi}$. These variables generalize the corresponding variables from the $k$-FDPC problem. The idea is to impose that $d_{vi}=0$ if and only if node $v$ is \emph{not selected} in walk $W_i$. 
Thanks to the conservation-of-flow condition \eqref{eqn:flow-conservation-walk}, we can model this as:

\begin{subequations}
\label{reach1}
\begin{align}
& d_{si} = 1 &&  \forall i \in \{1,\dots,k\},\\
& \sum_{(u,v) \in E} x_{uvi} = 0 \rightarrow d_{vi} = 0 && \forall v \in V \setminus \{s\}, \forall i \in \{1,\dots,k\}.
\end{align}
\end{subequations}

Moreover, for every selected node $v$ we want to guarantee that $v$ is reachable from $s$ using selected edges. This holds if and only if there is some selected edge $(u,v)$ such that $u$ is reachable from $s$. We will thus impose that every selected node $v$ has a least one in-coming selected edge $(u,v)$ such that $d_{vi} \geq d_{ui} + 1$ (these will correspond to a spanning tree of the selected nodes rooted at the source). 
We encode such selected in-coming edges giving the strict increase with a binary variable $y_{uvi}$ for every edge. Variable $y_{uvi}$ equals 1 for exactly one selected in-coming edge to each selected node $v \in V \setminus \{s\}$. We thus state the following constraints (see \Cref{fig:reachability-example} for an example):

\begin{subequations}
\begin{align}
& y_{uvi} = 1 \rightarrow x_{uvi} \geq 1 && \forall (u,v) \in E, \forall i \in \{1,\dots,k\},\label{eq:reach-1}\\
& \left(\sum_{(u,v) \in E} x_{uvi} \geq 1\right) \rightarrow \left(\sum_{(u,v) \in E} y_{uvi} = 1\right) && \forall v \in V \setminus \{s\},\forall i \in \{1,\dots,k\},\label{eq:reach-2}\\
& \left(\sum_{(u,v) \in E} x_{uvi} \geq 1\right) \rightarrow \left(\sum_{(u,v) \in E} y_{uvi}(d_{vi} - d_{ui}) \geq 1\right) && \forall v \in V \setminus \{s\}, \forall i \in \{1,\dots,k\}.\label{eq:reach-3}
\end{align}
\label{reach2}
\end{subequations}

The following lemma states that \cref{eqn:flow-conservation-walk,reach1,reach2} correctly model $s$-$t$ walks.

\begin{figure}
    \centering
    \includegraphics[scale=0.20]{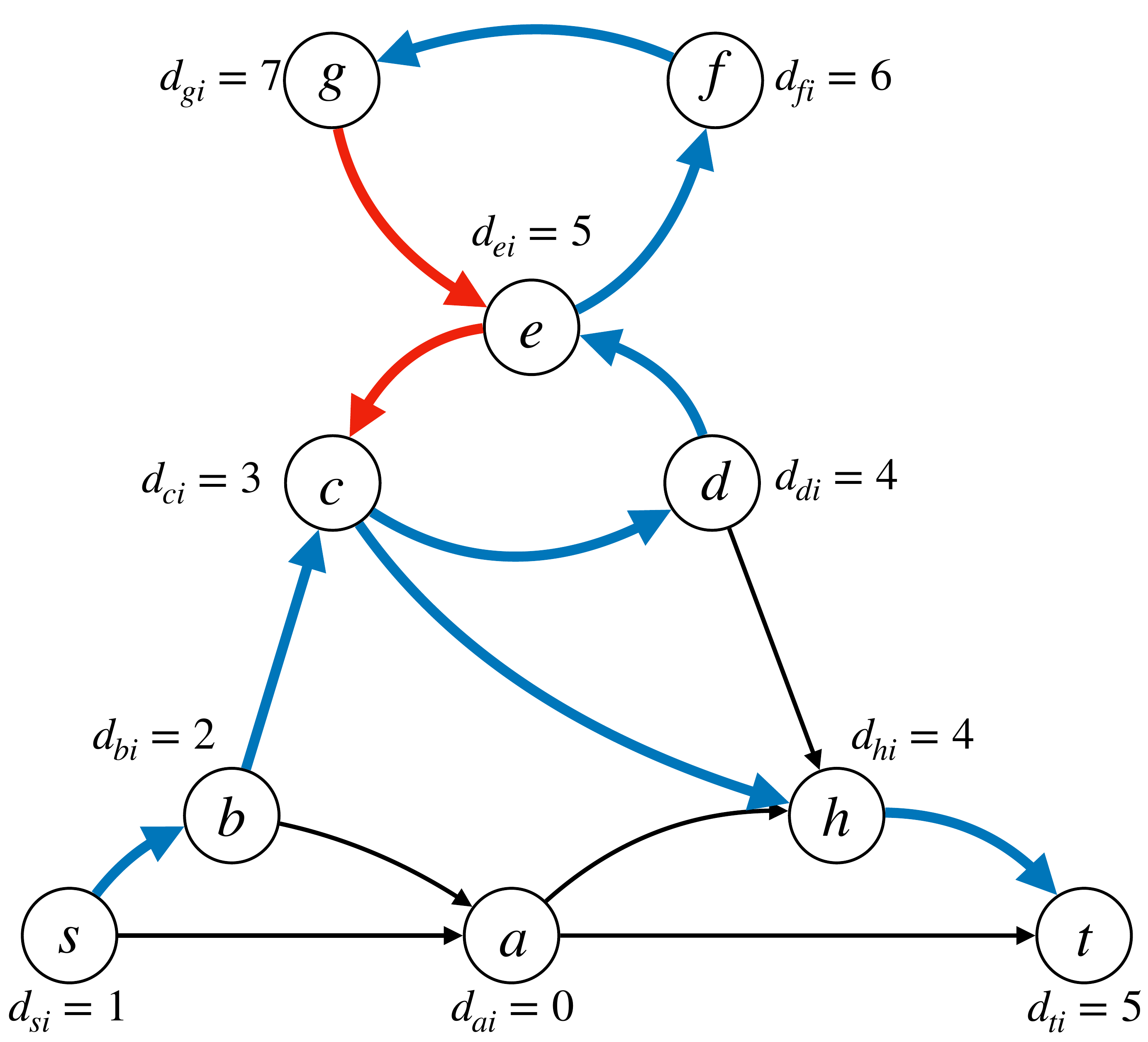}
    \caption{Illustration of \Cref{lem:walks-correctness}. Thick edges $(u,v)$ are those for which $x_{uvi} = 1$. An edge $(u,v)$ is drawn in blue if $y_{uvi} = 1$ (these form a spanning tree of the selected nodes rooted at $s$), and in red otherwise. Next to each node $v$, we draw the value of $d_{vi}$. Note that any blue edge is also thick (condition \eqref{eq:reach-1}), any selected node $v$ has exactly one incoming blue edge (condition \eqref{eq:reach-2}), and for any selected node $v$ its in-coming blue edge $(u,v)$ satisfies $d_{ui} < d_{vi}$ (condition \eqref{eq:reach-3}).
    \label{fig:reachability-example}
    }
\end{figure}

\begin{lemma}
\label{lem:walks-correctness}
The ILP made up of equations \eqref{eqn:flow-conservation-walk}, \eqref{reach1} and \eqref{reach2}, is feasible if and only if conditions 1 and 2 from Lemma \ref{lem:walks-reachability} hold.
\end{lemma}

\begin{proof}
For the forward direction, for every selected $v$, the condition in Eq. \eqref{eq:reach-3} guarantees that we have a selected edge $(u,v)$, such that $d_{vi} \geq d_{ui} + 1$. Thus, following all such edges back until $s$ we obtain an $s$-$v$ path.

For the reverse direction, for every $v \in V$, set $d_{vi}$ to be the number of nodes of a shortest $s$-$v$ path using only edges $(u,v)$ for which $x_{uvi} \geq 1$, or 0 if there is no such path.\qed
\end{proof}

Next, we explain how to transform the conditional structures in constraints \eqref{reach1} and \eqref{reach2}. The second condition in \eqref{reach1} can be expressed as:

\begin{align}
& \sum_{(u,v) \in E} x_{uvi} \geq d_{vi} && \forall v \in V \setminus \{s\}, \forall i \in \{1,\dots,k\}.
\end{align}

Similarly, the first and third conditions in \eqref{reach2} can be expressed as
\begin{align}
&   x_{uvi} \geq y_{uvi} && \forall (u,v) \in E, \forall i \in \{1,\dots,k\},\label{eq:walks-lin-2}\\
&  \sum_{(u,v) \in E} x_{uvi} \leq \sum_{(u,v) \in E} y_{uvi}(d_{vi} - d_{ui}) && \forall  v \in V \setminus \{s\}, \forall i \in \{1,\dots,k\}.\label{eq:walks-lin-3}
\end{align}

For the second condition in \eqref{reach2} we can apply a similar, but slightly modified, approach using the fact that $\sum_{(u,v) \in E} x_{uvi} \leq |E|$:
\begin{subequations}
\begin{align}
&   \sum_{(u,v) \in E} x_{uvi} \leq |E| \sum_{(u,v) \in E} y_{uvi} && \forall v \in V \setminus \{s\}, \forall i \in \{1,\dots,k\},\label{eq:walks-lin-2a}\\
&   \sum_{(u,v) \in E} y_{uvi} \leq 1 && \forall v \in V \setminus \{s\}, \forall i \in \{1,\dots,k\}.\label{eq:walks-lin-2b}
\end{align}
\end{subequations}

\begin{remark}
{
In order to linearize \Cref{eq:walks-lin-3}, note first that $y_{uvi}$ is a binary variable. Next, we can introduce a new integer variable $d_{uvi} = d_{vi} - d_{ui}$. Note that $n$ is an upper bound for any $d_{vi}$ (since the interval $[0,n]$ for each $d_{vi}$ variable is sufficiently large to allow for a feasible solution, if there is any), and thus we have that each $d_{uvi} \in [-n,n]$. As such, the product $y_{uvi}d_{uvi}$ can be linearized as in \Cref{remark:linearize-trails-elim} (except that now we have a negative lower bound).}

\begin{subequations}
\begin{align}
& \phi_{uvi} \leq M y_{uvi} && \forall (u,v) \in E, \forall i \in \{1,\dots,k\},\label{eq:walk_rep_2}\\
& \phi_{uvi} \geq -M y_{uvi} && \forall (u,v) \in E, \forall i \in \{1,\dots,k\},\label{eq:walk_rep_5}\\
& \phi_{uvi} \leq (d_{vi} - d_{ui}) + (1-y_{uvi})M && \forall (u,v) \in E, \forall i \in \{1,\dots,k\},\\
& \phi_{uvi} \geq (d_{vi} - d_{ui}) - (1-y_{uvi})M && \forall (u,v) \in E, \forall  i \in \{1,\dots,k\}.\label{eq:walk_rep_4}
\end{align}
\label{int_lin}
\end{subequations}
\end{remark}

\begin{remark}
\label{remark:two_comp}
{For walks, the flow superposition constraint in \Cref{eqn:ilp_flow_eq} involves the product of two integer variables (since, for walks, $x_{uvi}$ is not assumed to be binary). Thus, we cannot simply linearize this product using \Cref{rem:lin}, since that requires one of the variables to be binary. However, we can reduce this more general case to several applications of \Cref{rem:lin}, using a \emph{power-of-two technique} (see e.g.~\citep{koren2018computer}). In this technique, one of the integer variables is replaced by its expression as a sum of powers of two, by introducing an auxiliary binary variable for each power of two (smaller than the maximum value it can achieve).
More specifically, suppose we want to linearize the product $xy$ between two integer variables $x$ and $y$. Suppose also that $x \in [\underline{x},\overline{x}]$, with $0 \leq \underline{x}$. For each $j \in \{0,\dots,\lfloor\log_2(\overline{x})\rfloor\}$, we introduce a binary variable $x_j$ and add the constraint
\[x = \sum_{j \in \{0,\dots,\lfloor\log_2(\overline{x})\rfloor\}}2^jx_j.\]
The product $xy$ then becomes
\[\sum_{j \in \{0,\dots,\lfloor\log_2(\overline{x})\rfloor\}}2^jx_jy.\]
Each product $x_jy$ is now between the binary variable $x_j$ and the integer variable $y$, which can now be linearized as in \Cref{rem:lin}.}
\end{remark}

The complete ILP formulation for this problem variant, which we give as Model~\ref{mod:pw} in \ref{sec:k-fdw}, has $O((|V|+|E|)k)$ variables and constraints.

\begin{remark}
In order to obtain a $k$-flow decomposition into \emph{exactly} $k$ $s$-$t$ walk, we can make the same change as in \Cref{rem:trails-exactly-k}, by imposing the basic flow conservation constraint:
\begin{equation}
\label{eqn:flow-conservation-wallk2}
    \sum_{(u,v) \in E} x_{uvi} - \sum_{(v,u) \in E} x_{vui} = 
    \begin{cases}
    0, & \text{if $v \in V \setminus \{s,t\}$}, \forall i \in \{1,\dots,k\}\\
    1, & \text{if $v = t$}, \forall i \in \{1,\dots,k\}\\
    -1, & \text{if $v = s$}, \forall i \in \{1,\dots,k\}.
    \end{cases}
\end{equation}

The value of the constant $M$ should be large enough to guarantee that the left side of the constraints in \eqref{int_lin} are satisfied adequately. Hence, in constraint \eqref{eq:walk_rep_2}, when $y_{uvi}=1$, the maximum value for $\phi_{uvi}$ should be equal to the maximum value for $d_{uvi}$ (i.e. $\overline{d}_{uvi}$). In constraint \eqref{eq:walk_rep_4}, when $y_{uvi}=1$, $\phi_{uvi}$ should be at most lower bounded by zero and that can be achieve by using $M = \overline{d}_{uvi}$.

For walks, where $y_{uvi}$ and $x_{uvi}$ behaves as integer variables, this product can be linearized as in \Cref{remark:two_comp}. 
\end{remark}

\section{Experiments}
\label{exp}

\subsection{Experiment Design}
\label{sec:design}

\subsubsection{Solvers} We designed our experiments to test the minimization versions of the three problem variants from \Cref{def:problem-variants}. For \PC we implemented the model from Section \eqref{pc} (in full as Model~\ref{mod:pc} in \ref{sec:k-fdpc}). For \PT, we implemented both the iterative constraint generation approach described in \Cref{pt} (in full as Model~\ref{mod:pt} in \ref{sec:k-fdt}) and the model from \Cref{pw} (in full as Model~\ref{mod:pw} in \ref{sec:k-fdw}) with binary $x_{uvi}$ variables. For \PW we implemented the model from \Cref{pw} (in full as Model~\ref{mod:pw} in \ref{sec:k-fdw}) with $x_{uvi}$ as integer variables. To find the minimum $k$, we implemented the iterative search over all values of $k$ in increasing order, described at the end of \Cref{sec:network-flows-and-decompositions}. All four were implemented using the Python API for CPLEX 20.1 under default settings. We ran our experiments on a personal computer with 16 GB of RAM and an Apple M1 processor at 2.9 GHz.

\subsubsection{Datasets} We test the performance of the solvers under a range of biological and transportation graph topologies and flow weights, which we also make available at \url{github.com/algbio/MFD-ILP}. As our first dataset, we took one of the larger datasets produced by \cite{shao2017accurate} (\path{rnaseq/sparse_quant_SRR020730.graph}) and used in a number of flow decomposition benchmarking studies~\citep{kloster2018practical,williams2021flow}. This contains transcriptomic data as a flow in a DAG for each gene of the human genome. A slight alteration was applied to this dataset in order to create cycles in each instance. We call this dataset \textbf{SRR020730-Salmon-Adapted}.

For the second dataset, we used a collection of different instances from flow network transportation data curated by \cite{bar2021transportation}. We refer to this as \textbf{Transportation Data}. The dataset comprises network flows from different cities worldwide collected through various scenarios.

To obtain a more complex dataset, we created a dataset consisting of genome graphs from up to 50 variants of \emph{E. coli} genomes, with flow values coming from the abundances of these genomes. The complete dataset is composed of more than a thousand instances with a different number of genomes. For the sake of this experiment, we sampled up to 10 graphs from each number of genomes, starting from 2. The graphs in this dataset are the largest.

In our opinion, these three datasets cover a range of real-world applications, indicative of our models' efficiency and scalability in such scenarios. Since the precise details of these datasets are outside this paper's scope, we provide a more detailed description of how the datasets were created at \url{github.com/algbio/MFD-ILP}.

\subsection{Results and Discussion}
\label{sec:results}

\Cref{fig:plots} summarizes the performance of the models in terms of runtime, by showing the proportion of all instances that can be (individually) solved within a certain time. Overall, we observe that all four formulations are efficient, finishing in under {45} seconds on the first dataset, {75} seconds on the second dataset, and {1.5} minutes on the third dataset. Note that the datasets are of increasing complexity in terms of number of nodes and edges. Moreover, the runtime of \PC is relatively smaller than the runtime for \PT and \PW. This is expected given that the number of constraints and variables is smaller when handling paths and cycles.

\begin{figure}	
	\centering	
		\subfloat[\textbf{SRR020730 Salmon Adapted}]{\includegraphics[width=0.5\textwidth]{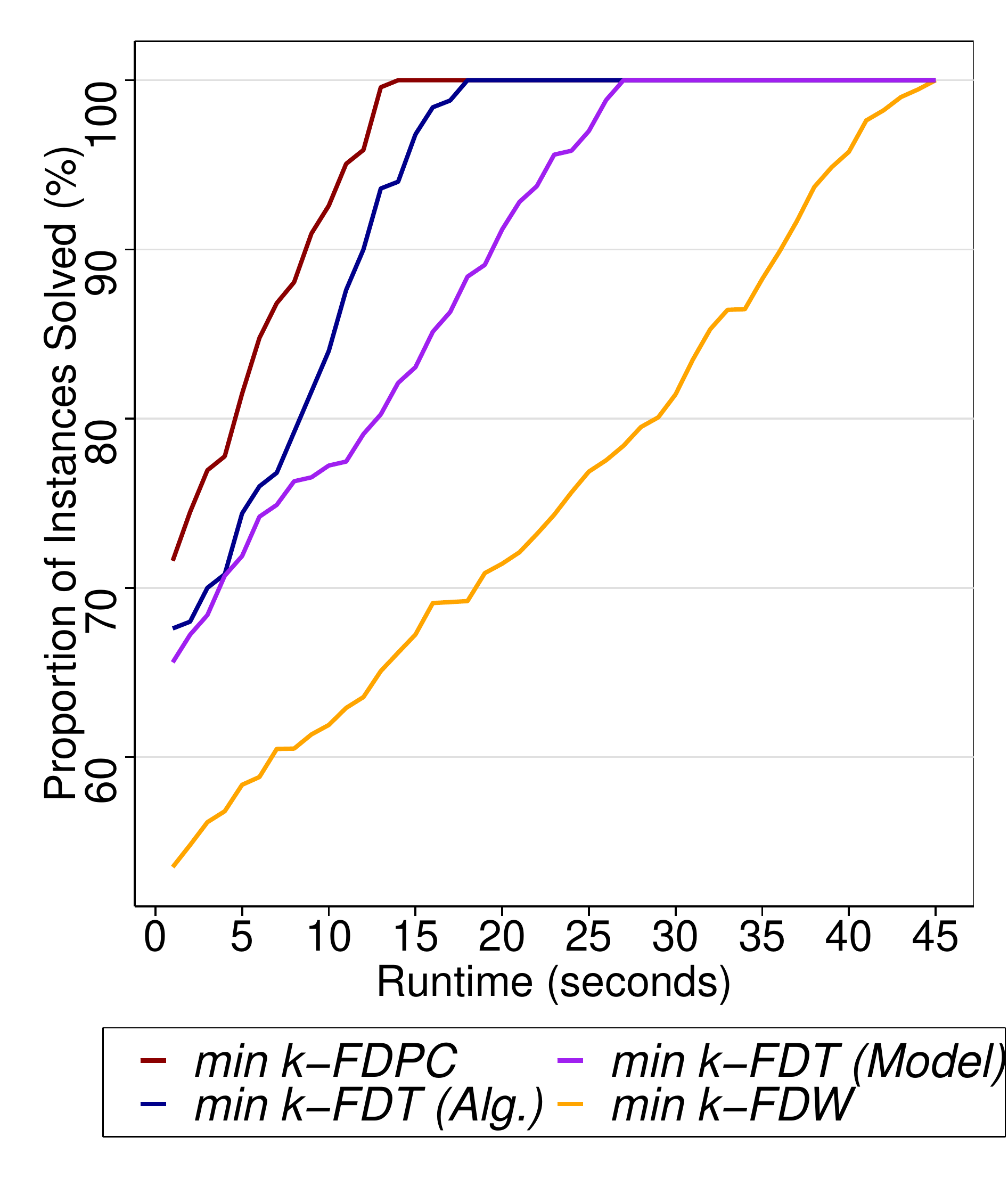}\label{fig:salmon}}
		\subfloat[\textbf{Transportation Data}]{\includegraphics[width=0.5\textwidth]{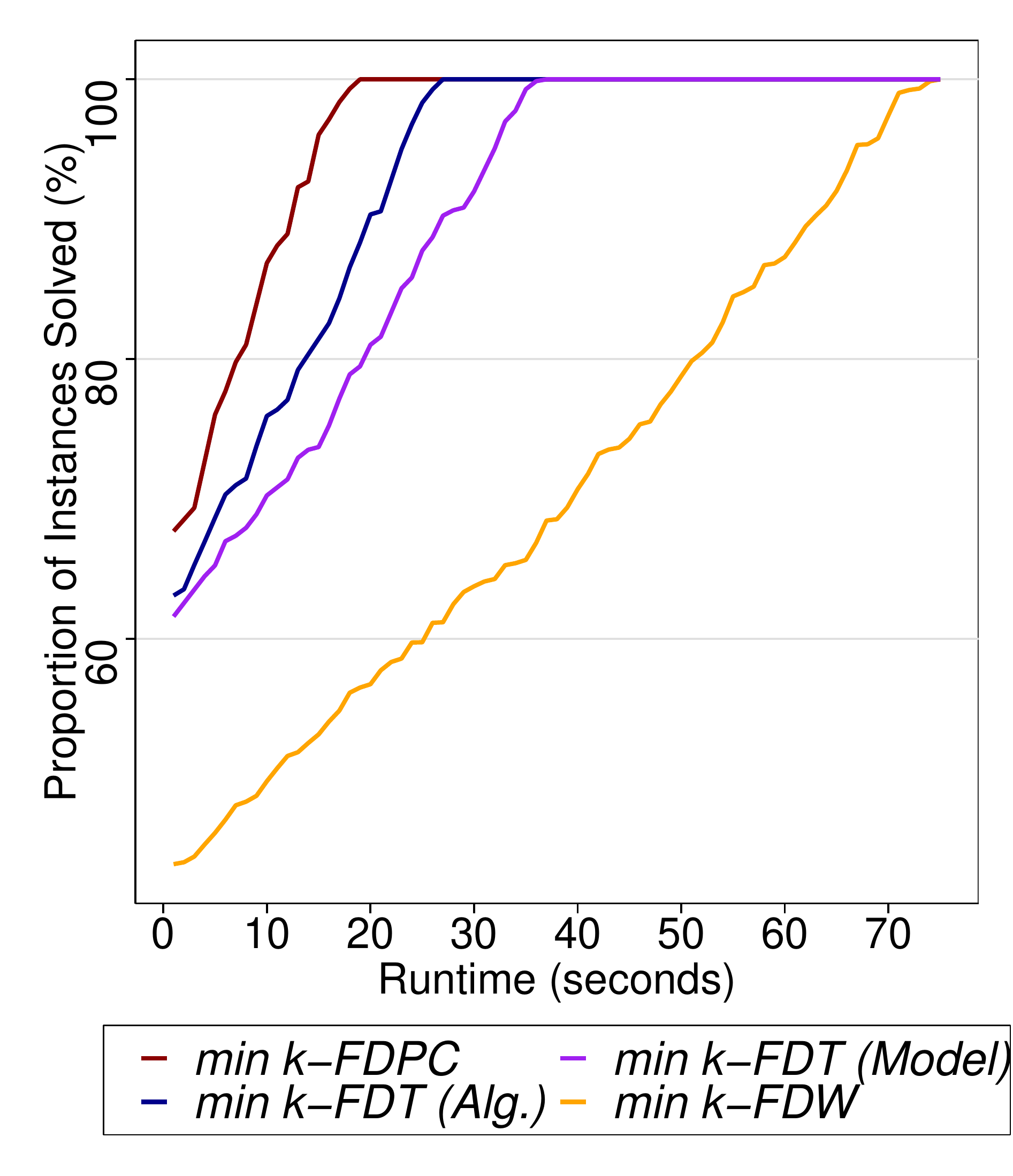} \label{fig:transport}}\\
		\subfloat[\textbf{E. coli Strains}]{\includegraphics[width=1\textwidth]
		{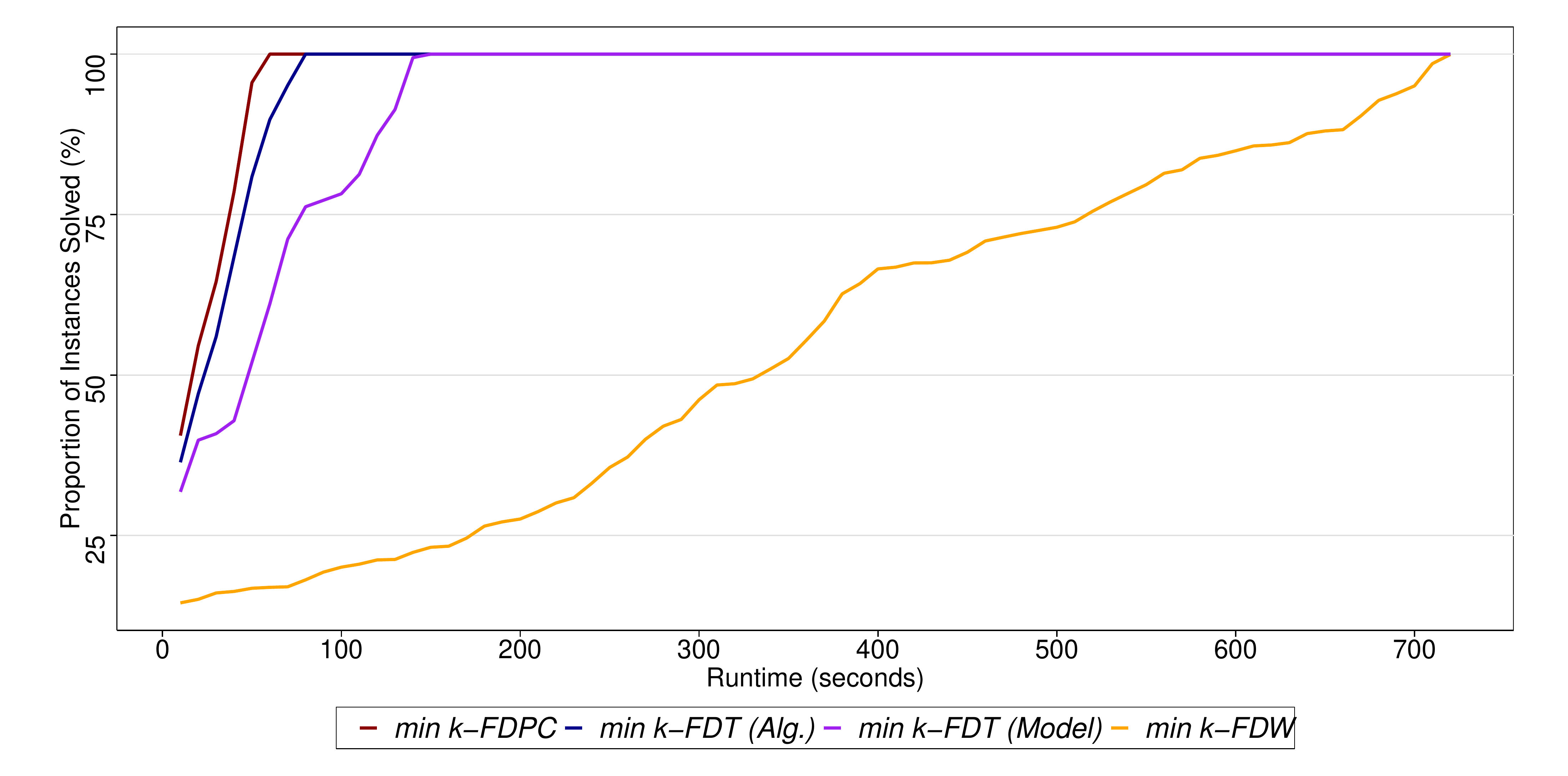} \label{fig:bacteria}}\\
	\caption{The percentage of all instances (y-axis) individually solvable within a certain number of seconds (x-axis).}
	\label{fig:plots}
\end{figure} 

We also evaluate how the models scale in terms of complexity of the instances of each dataset. We focus on the size $k$ of the flow decomposition since the models have size linear in $k$. In \Cref{tab:variants_runtime} we group the graphs into ranges based on the minimum size of a flow decomposition into paths or cycles since this is the most basic problem variant (column $min~k$). Since $|V|$ and $|E|$ also increase as $min~k$ increases, this table also indicates how the models scale in terms of graph sizes. The runtimes of all models do scale up adequately, reiterating that they are a viable candidate for real applications. The reachability model for trails (where its main decision variables $x_{uvi}$ are binary) scales better than for walks (where these variables are integer).

\begin{table}
\centering
\caption{Averages and standard deviations (in parentheses) for graph size, running times of the four variants of problem, and the number of iterations of \Cref{algo:trail} until a feasible solution is found. For \Cref{algo:trail} the runtime is the total over all iterations. If an instance did not admit a decomposition into trails, we do not include it in the runtime average for $\emph{min}$ $k$-FDT.}
\label{tab:variants_runtime}
\resizebox{1.0\columnwidth}{!}{%
\begin{tabular}{l lcll  l  lll  l}
\toprule
&&&&& \multicolumn{1}{c}{\multirow{2}{*}{\parbox{2.2cm}{\centering\smallskip $\emph{min}$ $k$-FDPC\\ (sec) \smallskip}}} &\multicolumn{3}{c}{$\emph{min}$ $k$-FDT (sec)} & \multicolumn{1}{c}{\multirow{2}{*}{\parbox{2.2cm}{\centering\smallskip $\emph{min}$ $k$-FDW\\ (sec) \smallskip}}}\\
\cmidrule(r){7-9}  
 & \emph{min} $k$ & $\#$instances & $|V|$ & $|E|$ & & \cref{algo:trail} & $\#$iter & Model~\ref{mod:pw} & \\
 \midrule
\multirow{7}{*}{\rotatebox{90}{\shortstack{\textbf{SRR020730}\\\textbf{Salmon}\\\textbf{Adapted}}}}
\\
& 4-10	        & 5131  &	23.42 (4.21)	&	45.12  (8.22) & 3.22  (0.36) &  5.13  (0.52) & 2.32 & {10.2 (1.41)}  & {13.42 (1.12)} \\
& 11-15	        & 595    &	32.34 (3.42)	&	60.21  (5.41) & 4.43  (0.93) &  9.13 (1.23) & 2.94 & 14.5 (1.54)  & {19.88 (3.41)}\\
& 16-20	        & 236	& 	43.81 (4.67)	&	71.14  (9.57) & 8.68  (1.72) &  14.03 (2.12) & 3.20 & {17.2 (2.72)}  & {28.61 (5.74)} \\
& 21-\emph{max}	& 60    &   51.42 (5.23)	&	102.14 (11.31) & 12.42 (2.41) & 18.12 (3.34) & 4.51 & {27.2 (4.41)}  & {39.79 (5.81)} \\
\\
\midrule
\multirow{6}{*}{\rotatebox{90}{\shortstack{\textbf{Transportation}\\\textbf{Data}}}}
\\
&4-10	        & 200   &	43.22  (6.21)	&	78.21 (10.13) &  4.51  (0.56) & 8.32 (1.13)  & 2.41 & {16.54 (2.12)} & {28.22 (3.31)} \\
&11-15	        & 150   &	50.74  (7.51)	&	83.21  (7.65) &  7.13  (1.32) & 15.21 (2.41)  & 3.52 & {21.33 (3.23)} & {53.15 (6.13)} \\
&16-20	        & 50	& 	64.33  (4.78)	&	93.21  (5.11) & 10.42  (2.13) & 17.89 (3.29)  & 5.93 & {28.92 (4.09)} & {78.61 (7.21)} \\
&21-\emph{max}	& 20    &   80.32 (5.21)	&	130.1  (3.13) & 14.54  (1.21) & 24.72 (2.12)  & 6.32 & {36.94 (3.21)} & {112.71 (6.88)} \\
\\
\midrule
\multirow{6}{*}{\rotatebox{90}{\shortstack{\textbf{E. coli}\\\textbf{Strains}}}}
&4-10	        & 80    &	48.21    (5.41)	&	82.13  (7.24)  & 8.51   (2.41) & 14.13  (3.13)  &  4.33  &  {22.41 (5.31)}  &  {56.12  (18.25)} \\
&11-15	        & 50    &	63.73  (11.42)	&	90.22  (12.31) & 9.24   (3.13) & 17.31  (4.41)  &  5.92  &  {35.38 (10.24)}  & {147.21  (27.13)} \\
&16-20	        & 50	& 	78.34  (22.14)	&	108.01 (15.44) & 17.23  (4.62) & 25.89  (7.13)  &  6.24  &  {53.33 (14.41)}  & {297.35  (42.34)} \\
&21-25      	& 50    &   91.41  (28.13)	&	140.13 (17.62) & 24.46  (8.42) & 30.43  (9.31)  &  7.51  &  {61.42 (22.14)} & {412.32  (58.42)} \\
&26-30          & 50    &   100.81 (33.32)  &   156.32 (24.54) & 36.93 (10.13) & 45.13 (13.41)  &  8.43  &  {78.43 (29.43)} & {523.52  (58.21)}\\
&31-\emph{max}         & 75    &   120.42 (30.34)  &   176.42 (30.32) & 60.61 (12.42) & 79.41 (18.31)  & 10.21  & {144.31 (34.62)} & {657.64  (63.55)}\\
\bottomrule
%\multirow{4}{*}{\rotatebox{90}{\StT}}
% \bottomrule
\end{tabular}
}
\end{table}

Notice that the scalability in terms of $min~k$ was an issue for MFD algorithms on DAGs before the ILP of \cite{dias2022fast} because the previous exact FPT algorithm of \cite{kloster2018practical} had an exponential dependency on $min~k$. Our results imply that also for cyclic graphs, ILP-based solutions have no issues in scaling to larger values of $min~k$. Thus, even if an FPT algorithm for MFD in graphs with cycles (parameterized by the minimum size of a solution) is open, our results suggest that such an FPT algorithm might not be necessary for practice.

To compare the three problem variants, in \Cref{tab:variantes_numbers} we study each decomposition type's minimum size. Since a decomposition with trails does not always exist, we also show the percentage of instances admitting such a decomposition. While all instances for all three datasets always admit a decomposition into trails for $min~k \leq 15$, for larger $min~k$ values, up to 10\%, 5\%, and 24\% of the instances in the first, second and third datasets, respectively, do not admit such a decomposition. On average, we observe that the minimum solution size for \PC is more significant than for \PT, which is more prominent for \PW. This is explained by the fact that a trail may use a node multiple times, and a walk may also use an edge numerous times. As such, a trail or a walk may combine a source-to-sink path with one or more cycles if their weights can be appropriately adjusted to form a flow decomposition. All these insights show that there is a clear difference between the problem (as also discussed in \Cref{into}), and that all the different variants may be explored in applications.

\begin{table}
\centering
 \caption{Average minimum solution sizes for the flow decomposition problem into path or cycles (also split into the average number of paths and cycles), into trails (and the percentage of instances admitting a flow decomposition into trails), and into walks.
 }.
\label{tab:variantes_numbers}
\resizebox{0.8\columnwidth}{!}{%
\begin{tabular}{r r rrc rr r}
&& \multicolumn{3}{c}{$\emph{min}$ $k$-FDPC } &\multicolumn{2}{c}{$\emph{min}$ $k$-FDT } & \multicolumn{1}{c}{$\emph{min}$ $k$-FDW}\\
\cmidrule(r){3-5}   \cmidrule(r){6-7}   \cmidrule(r){8-8}  
& \emph{min} $k$ & $\#$paths & $\# $cycles &$\#$path or cycles & $\#$trails &$\%$feasible & $\#$walks \\
\midrule
\multirow{7}{*}{\rotatebox{90}{\shortstack{\textbf{SRR020730}\\\textbf{Salmon}\\\textbf{Adapted}}}}
\\
&4-10	        & 5.1   & 4.4  & 9.5 & 3.5  & 100  & 2.8 \\
&11-15	        & 7.1   & 6.2  & 13.3 & 5.1  & 100  & 4.3 \\
&16-20	        & 9.7   & 8.1  & 17.7 & {5.8}  &  93  & {5.4} \\
&21-\emph{max}	& 13.1  & 10.3 & 23.3 & {10.1} &  90  & {9.8} \\
\\
\midrule
\multirow{6}{*}{\rotatebox{90}{\shortstack{\textbf{Transportation}\\\textbf{Data}}}}
\\
&4-10	        & 5.4   & 3.2 &  8.6  &  {4.4} & 100 & {3.4} \\
&11-15	        & 6.1   & 4.4 &  10.5  &  {5.9} & 100 & {4.8} \\
&16-20	        & 8.2   & 7.3 & 15.5  &  {7.3} &  98 & {5.4} \\
&21-\emph{max}	& 10.1  & 8.3 & 18.4  &  {8.8} &  95 & {4.8} \\
\\
\midrule
\multirow{6}{*}{\rotatebox{90}{\shortstack{\textbf{E. coli}\\\textbf{Strains}}}}
& 4-10          & 5.3  & 4.9  & 10.2  &  {4.5} & 100 &  {3.1}\\
& 11-15         & 8.9  & 4.4  & 13.2 &  {8.2} & 100 &  {6.8} \\
& 16-20         & 10.5 & 6.3  & 16.8 &  {9.1} & 100 &  {7.8} \\
& 21-25         & 15.2 & 8.5  & 23.7 & {12.6} & 90  &  {9.9} \\
& 26-30         & 19.3 & 7.1  & 26.4 & {15.1} & 84  &  {12.4} \\
& 31-\emph{max}	& 23.1 & 11.4 & 34.5 & {20.1} & 76  &  {19.6} \\
\bottomrule
\end{tabular}
}
\end{table}

Finally, we analyze the two solutions for \PT, constraint generation (Model~\ref{mod:pt}) and Model~\ref{mod:pw} with variables $x_{uvi}$ binary. While both are efficient, the former is generally faster than the latter because it has a smaller number of constraints, and on our data, it also requires a small number of iterations (at most 11). However, there is no guarantee that it will be concluded after a small number of iterations in general (but as mentioned in \Cref{pt}, it does find a feasible decomposition after a finite number of iterations). Another advantage provided by the constraint generation formulation is that algorithmic verification procedures can assist it (in the same way as we delegated the strong connectivity check to an algorithm) and use this to craft more straightforward constraints in the original model (in our case, simple constraints based on out-going edges). 

\section{Conclusions and Future Work}
\label{con}

Flow decomposition is a common problem present in different fields of science, including Computer Science, Bioinformatics and Transportation. 
Throughout the literature, there has been considerable effort into solving this problem, especially on acyclic inputs. This is mainly because acyclicity guarantees strong properties that can be used to develop algorithms (fixed-parameter tractable, approximation). However, for graphs that contain cycles, there has yet to be an exact solution proposed so far, the few current approaches being heuristics based on greedy algorithms.

This paper considers three natural variants of the flow decomposition problem in graphs with cycles: decompositions into paths or cycles, trails, and walks, respectively. Our ILP formulations generally adopt the same strategy as in the acyclic case from~\cite{dias2022fast}, namely formulating what constitutes an element of decomposition and then requiring that these weighted elements fit the input flow values. However, the novelty of our formulations resides in modelling the different types of walks in a cyclic graph, which is notably more involved than formulating a path in a DAG. These formulations can also be of independent interest outside the flow decomposition problem since they model basic graph-theoretic notions.

Our formulations are also extensively tested on biological and transportation datasets. Despite the problems being NP-hard, they solve any instance of the three datasets under {45 seconds, 75 seconds, and 12 minutes}, respectively. Nevertheless, since these are the first exact solutions to the problem on cyclic graphs, the quest for more efficient solutions remains open.

As future work, it would be interesting to extend the models presented here to include other aspects of empirical data, such as all the flow decomposition problem variants discussed in~\cite{dias2022fast} for Bioinformatics applications. More specifically, in one problem variant, we are also given a set of paths (called \emph{subpath constraints}) that must appear in at least one walk of the flow decomposition. The solution from~\cite{dias2022fast} for this variant can be easily adapted to the \PC problem (because edges are not repeated), but it remains open how to adapt it to trails or walks because they can repeat edges. In the \emph{inexact} and \emph{imperfect} flow decomposition variants from~\cite{dias2022fast}, the weighted walks do not need to fit the input flow perfectly. However, the constraints modelling such solutions from~\cite{dias2022fast} immediately carry over to all our problem variants. Lastly, our reachability formulation can also handle problem variants in which the solution walks can pass through an edge at most a given number of times by just setting a bound to each $x_{uvi}$ variable.

\section*{Acknowledgments}
This work was partially funded by the European Research Council (ERC) under the European Union's Horizon 2020 research and innovation programme (grant agreement No.~851093, SAFEBIO), partly by the Academy of Finland (grants No.~322595, 328877, 346968), and partially by the US NSF (award 1759522).

\bibliographystyle{apalike}
\bibliography{Past/main}

\begin{thebibliography}{}

\bibitem[Ahuja et~al., 1988]{ahuja1988network}
Ahuja, R.~K., Magnanti, T.~L., and Orlin, J.~B. (1988).
\newblock Network flows.

\bibitem[Baaijens et~al., 2020]{baaijens2020strain}
Baaijens, J.~A., Stougie, L., and Sch{\"o}nhuth, A. (2020).
\newblock Strain-aware assembly of genomes from mixed samples using flow
  variation graphs.
\newblock In {\em International Conference on Research in Computational
  Molecular Biology}, pages 221--222. Springer.

\bibitem[Bar-Gera et~al., 2021]{bar2021transportation}
Bar-Gera, H., Stabler, B., and Sall, E. (2021).
\newblock Transportation networks for research core team.
\newblock {\em Transportation Network Test Problems. Available online:
  \url{https://github.com/bstabler/TransportationNetworks} (accessed on 3 July
  2022)}.

\bibitem[Bentley and Yao, 1976]{bentley1976almost}
Bentley, J.~L. and Yao, A. C.-C. (1976).
\newblock An almost optimal algorithm for unbounded searching.
\newblock {\em Information Processing Letters}, 5(3):82--87.

\bibitem[C{\'a}ceres et~al., 2022]{Caceres:2022vi}
C{\'a}ceres, M., Cairo, M., Grigorjew, A., Khan, S., Mumey, B., Rizzi, R.,
  Tomescu, A.~I., and Williams, L. (2022).
\newblock Width helps and hinders splitting flows.
\newblock In {\em {ESA 2022 - European Symposium on Algorithms}}.
\newblock To appear, \url{https://arxiv.org/abs/2207.02136}.

\bibitem[Cohen et~al., 2014]{cohen2014effect}
Cohen, R., Lewin-Eytan, L., Naor, J.~S., and Raz, D. (2014).
\newblock {On the effect of forwarding table size on SDN network utilization}.
\newblock In {\em IEEE INFOCOM 2014-IEEE conference on computer
  communications}, pages 1734--1742. IEEE.

\bibitem[Dantzig et~al., 1954]{dantzig1954solution}
Dantzig, G., Fulkerson, R., and Johnson, S. (1954).
\newblock Solution of a large-scale traveling-salesman problem.
\newblock {\em Journal of the operations research society of America},
  2(4):393--410.

\bibitem[Dias et~al., 2022]{dias2022fast}
Dias, F. H.~C., Williams, L., Mumey, B., and Tomescu, A.~I. (2022).
\newblock {Fast, Flexible, and Exact Minimum Flow Decompositions via ILP}.
\newblock In {\em RECOMB 2022 - 26th Annual International Conference on
  Research in Computational Molecular Biology}, volume 13278 of {\em Lecture
  Notes in Computer Science}, pages 230--245. Springer.

\bibitem[Grabherr et~al., 2011]{grabherr2011trinity}
Grabherr, M.~G., Haas, B.~J., Yassour, M., Levin, J.~Z., Thompson, D.~A., Amit,
  I., Adiconis, X., Fan, L., Raychowdhury, R., Zeng, Q., et~al. (2011).
\newblock {Trinity: reconstructing a full-length transcriptome without a genome
  from RNA-Seq data}.
\newblock {\em Nature biotechnology}, 29(7):644.

\bibitem[{Gurobi Optimization, LLC}, 2021]{gurobi}
{Gurobi Optimization, LLC} (2021).
\newblock {Gurobi Optimizer Reference Manual}.

\bibitem[Hartman et~al., 2012]{hartman2012split}
Hartman, T., Hassidim, A., Kaplan, H., Raz, D., and Segalov, M. (2012).
\newblock How to split a flow?
\newblock In {\em 2012 Proceedings IEEE INFOCOM}, pages 828--836. IEEE.

\bibitem[Hong et~al., 2013]{hong2013achieving}
Hong, C.-Y., Kandula, S., Mahajan, R., Zhang, M., Gill, V., Nanduri, M., and
  Wattenhofer, R. (2013).
\newblock Achieving high utilization with software-driven wan.
\newblock In {\em Proceedings of the ACM SIGCOMM 2013 conference on SIGCOMM},
  pages 15--26.

\bibitem[Kloster et~al., 2018]{kloster2018practical}
Kloster, K., Kuinke, P., O'Brien, M.~P., Reidl, F., Villaamil, F.~S., Sullivan,
  B.~D., and van~der Poel, A. (2018).
\newblock A practical fpt algorithm for flow decomposition and transcript
  assembly.
\newblock In {\em 2018 Proceedings of the Twentieth Workshop on Algorithm
  Engineering and Experiments (ALENEX)}, pages 75--86. SIAM.

\bibitem[Koren, 2018]{koren2018computer}
Koren, I. (2018).
\newblock {\em Computer arithmetic algorithms}.
\newblock AK Peters/CRC Press.

\bibitem[Kovaka et~al., 2019]{kovaka2019transcriptome}
Kovaka, S., Zimin, A.~V., Pertea, G.~M., Razaghi, R., Salzberg, S.~L., and
  Pertea, M. (2019).
\newblock {Transcriptome assembly from long-read RNA-seq alignments with
  StringTie2}.
\newblock {\em Genome biology}, 20(1):1--13.

\bibitem[Lin et~al., 2012]{lin2012cliiq}
Lin, Y.-Y., Dao, P., Hach, F., Bakhshi, M., Mo, F., Lapuk, A., Collins, C., and
  Sahinalp, S.~C. (2012).
\newblock Cliiq: Accurate comparative detection and quantification of expressed
  isoforms in a population.
\newblock In {\em International Workshop on Algorithms in Bioinformatics},
  pages 178--189. Springer.

\bibitem[Miller et~al., 1960]{miller1960integer}
Miller, C.~E., Tucker, A.~W., and Zemlin, R.~A. (1960).
\newblock Integer programming formulation of traveling salesman problems.
\newblock {\em Journal of the ACM (JACM)}, 7(4):326--329.

\bibitem[Mumey et~al., 2015]{mumey2015parity}
Mumey, B., Shahmohammadi, S., McManus, K., and Yaw, S. (2015).
\newblock Parity balancing path flow decomposition and routing.
\newblock In {\em 2015 IEEE Globecom Workshops (GC Wkshps)}, pages 1--6. IEEE.

\bibitem[Ohst, 2015]{Ohst:2015aa}
Ohst, J.~P. (2015).
\newblock {\em On the Construction of Optimal Paths from Flows and the Analysis
  of Evacuation Scenarios}.
\newblock PhD thesis, University of Koblenz and Landau, Germany.

\bibitem[Olsen et~al., 2020]{olsen2020study}
Olsen, N., Kliewer, N., and Wolbeck, L. (2020).
\newblock A study on flow decomposition methods for scheduling of electric
  buses in public transport based on aggregated time--space network models.
\newblock {\em Central European Journal of Operations Research}, pages 1--37.

\bibitem[Safikhani et~al., 2013]{safikhani2013ssp}
Safikhani, Z., Sadeghi, M., Pezeshk, H., and Eslahchi, C. (2013).
\newblock {SSP: An interval integer linear programming for de novo
  transcriptome assembly and isoform discovery of RNA-seq reads}.
\newblock {\em Genomics}, 102(5-6):507--514.

\bibitem[Sashittal et~al., 2021]{jumper}
Sashittal, P., Zhang, C., Peng, J., and El-Kebir, M. (2021).
\newblock Jumper enables discontinuous transcript assembly in coronaviruses.
\newblock {\em Nature Communications}, 12(1):6728.

\bibitem[Schulz et~al., 2012]{schulz2012oases}
Schulz, M.~H., Zerbino, D.~R., Vingron, M., and Birney, E. (2012).
\newblock {Oases: robust de novo RNA-seq assembly across the dynamic range of
  expression levels}.
\newblock {\em Bioinformatics}, 28(8):1086--1092.

\bibitem[Shao and Kingsford, 2017a]{shao2017accurate}
Shao, M. and Kingsford, C. (2017a).
\newblock Accurate assembly of transcripts through phase-preserving graph
  decomposition.
\newblock {\em Nature biotechnology}, 35(12):1167--1169.

\bibitem[Shao and Kingsford, 2017b]{shao2017theory}
Shao, M. and Kingsford, C. (2017b).
\newblock Theory and a heuristic for the minimum path flow decomposition
  problem.
\newblock {\em IEEE/ACM transactions on computational biology and
  bioinformatics}, 16(2):658--670.

\bibitem[Soltysik and Yarnold, 2010]{soltysik2010two}
Soltysik, R.~C. and Yarnold, P.~R. (2010).
\newblock Two-group multioda: Mixed-integer linear programming solution with
  bounded $m$.
\newblock {\em Optimal Data Analysis}, 1(1):30--37.

\bibitem[Studio, 2017]{studio2017cplex}
Studio, I. I. C.~O. (2017).
\newblock Cplex users manual, version 12.7.

\bibitem[Taccari, 2016]{taccari2016integer}
Taccari, L. (2016).
\newblock Integer programming formulations for the elementary shortest path
  problem.
\newblock {\em European Journal of Operational Research}, 252(1):122--130.

\bibitem[Tarjan, 1972]{tarjan1972depth}
Tarjan, R. (1972).
\newblock Depth-first search and linear graph algorithms.
\newblock {\em SIAM journal on computing}, 1(2):146--160.

\bibitem[Vatinlen et~al., 2008]{VATINLEN20081390}
Vatinlen, B., Chauvet, F., Chr{\'e}tienne, P., and Mahey, P. (2008).
\newblock Simple bounds and greedy algorithms for decomposing a flow into a
  minimal set of paths.
\newblock {\em European Journal of Operational Research}, 185(3):1390--1401.

\bibitem[Williams et~al., 2019]{williams2019rna}
Williams, L., Reynolds, G., and Mumey, B. (2019).
\newblock {RNA Transcript Assembly Using Inexact Flows}.
\newblock In {\em 2019 IEEE International Conference on Bioinformatics and
  Biomedicine (BIBM)}, pages 1907--1914. IEEE.

\bibitem[Williams et~al., 2021]{williams2021flow}
Williams, L., Tomescu, A., Mumey, B.~M., et~al. (2021).
\newblock Flow decomposition with subpath constraints.
\newblock In {\em 21st International Workshop on Algorithms in Bioinformatics
  (WABI 2021)}. Schloss Dagstuhl-Leibniz-Zentrum f{\"u}r Informatik.

\bibitem[Zhang et~al., 2021]{scallop2}
Zhang, Q., Shi, Q., and Shao, M. (2021).
\newblock Scallop2 enables accurate assembly of multiple-end rna-seq data.
\newblock {\em bioRxiv}.

\bibitem[Zhao et~al., 2021]{zhao2021multitrans}
Zhao, J., Feng, H., Zhu, D., and Lin, Y. (2021).
\newblock Multitrans: an algorithm for path extraction through mixed integer
  linear programming for transcriptome assembly.
\newblock {\em IEEE/ACM Transactions on Computational Biology and
  Bioinformatics}.

\end{thebibliography}

\appendix
\newpage
\section{Full ILP formulation for $k$-Flow Decomposition into Paths or Cycles ($k$-FDPC)}
\label{sec:k-fdpc}

\begin{subequations}
\label{mod:pc}
\begin{align}
%& \sum_{(u,v) \in E} x_{uvi} \leq 1 && \forall u \in V, \forall i \in \{1, \ldots, k\}, \\
& \sum_{(u,v) \in E} x_{uvi} - \sum_{(v,w) \in E} x_{vwi} = 0 &&  \forall i \in \{1, \ldots, k\}, \forall v \in V \setminus \{s, t\},\\
& f_{uv} = \sum_{i \in \{1,\dots,k\}} \pi_{uvi} && \forall (u,v) \in E, \\
& \pi_{uvi} \leq f_{uv} x_{uvi} && \forall (u,v) \in E, \forall i \in \{1,\dots,k\},\\
& \pi_{uvi} \leq w_i && \forall (u,v) \in E, \forall i \in \{1,\dots,k\},\\
& \pi_{uvi} \geq w_i - (1-x_{uvi})f_{uv} && \forall (u,v) \in E, \forall i \in \{1,\dots,k\},\\
& d_{vi} \geq d_{ui} + 1 + (n-1)(x_{uvi} - 1) && \forall (u,v) \in E, \forall i \in \{1,\dots,k\},\\
& \sum_{v\in V} x_{svi} + \sum_{v\in V} c_{vi} \leq 1 && \forall i \in \{1,\dots,k\}, \\
& w_i \in \mathbb{Z}^+ && \forall i \in \{1,\dots,k\},\\
& x_{uvi} \in \{0,1\} && \forall (u,v) \in E, \forall i \in \{1,\dots,k\},\\
& \pi_{uvi} \in \mathbb{Z}^+ \cup \{0\} && \forall (u,v) \in E, \forall i \in \{1,\dots,k\},\\
& c_{vi} \in \{0,1\} && \forall v \in V, \forall i \in \{1,\dots,k\},\\
& d_{vi} \in \mathbb{Z}^+ && \forall v \in V, \forall i \in \{1,\dots,k\}.
\end{align}
\end{subequations}

\newpage
\section{Full ILP formulation for $k$-Flow Decomposition into Trails via Constraint Generation ($k$-FDT($\mathcal{C}$))}
\label{sec:k-fdt}

\begin{subequations}
\label{mod:pt}
\begin{align}
 & \sum_{(s,v) \in E} x_{svi} \leq 1 && \forall i \in \{1, \ldots, k\},\\
 & \sum_{(u,v) \in E} x_{uvi} - \sum_{(v,w) \in E} x_{vwi} = 0 &&  \forall i \in \{1, \ldots, k\}, \forall v \in V \setminus \{s, t\} \\
 & f_{uv} = \sum_{i \in \{1,\dots,k\}} \pi_{uvi} && \forall (u,v) \in E, \\
 & \pi_{uvi} \leq f_{uv} x_{uvi} && \forall (u,v) \in E, \forall i \in \{1,\dots,k\},\\
 & \pi_{uvi} \leq w_i && \forall (u,v) \in E, \forall i \in \{1,\dots,k\},\\
 & \pi_{uvi} \geq w_i - (1-x_{uvi})f_{uv} && \forall (u,v) \in E, \forall i \in \{1,\dots,k\},\\
 & \sum_{(u,v) \in E(C)} x_{uvi} \geq |C| - M(1-\beta_{Ci})  && \forall C \in \mathcal{C}, \forall i \in \{1,\dots,k\},\\
 & \sum_{(u,v) \in E(C)} x_{uvi} - |C| + 1 - M\beta_{Ci} \leq 0 && \forall C \in \mathcal{C}, \forall i \in \{1,\dots,k\},\\ 
 & \sum_{(u,v) \in \delta^+(C) \setminus E(C)} x_{uvi} \geq \beta_{Ci} && \forall C \in \mathcal{C}, \forall i \in \{1,\dots,k\} \\ 
 & w_i \in \mathbb{Z}^+ && \forall i \in \{1,\dots,k\},\\
 & x_{uvi} \in \{0,1\} && \forall (u,v) \in E, \forall i \in \{1,\dots,k\},\\
 & \pi_{uvi} \in \mathbb{Z}^+ \cup \{0\} && \forall (u,v) \in E, \forall i \in \{1,\dots,k\},\\
 & \beta_{Ci} \in \{0,1\}  && \forall C \in \mathcal{C}, \forall i \in \{1,\dots,k\}.
 \end{align}
 \end{subequations}

\section{Full ILP formulation for $k$-Flow Decomposition into Walks ($k$-FDW): $x_{uvi}$ as binary variables}
\label{sec:k-fdw}

\begin{subequations}
\label{mod:pw}
\begin{align}
 & \sum_{(s,v) \in E} x_{svi} \leq 1 && \forall i \in \{1, \ldots, k\}, \\
 & \sum_{(u,v) \in E} x_{uvi} - \sum_{(v,w) \in E} x_{vwi} = 0 &&  \forall i \in \{1, \ldots, k\}, \forall v \in V \setminus \{s, t\},\\
 & f_{uv} = \sum_{i \in \{1,\dots,k\}} \pi_{uvi} && \forall (u,v) \in E, \\
 & \pi_{uvi} \leq f_{uv} x_{uvi} && \forall (u,v) \in E, \forall i \in \{1,\dots,k\},\\
 & \pi_{uvi} \leq w_i && \forall (u,v) \in E, \forall i \in \{1,\dots,k\},\\
 & \pi_{uvi} \geq w_i - (1-x_{uvi})f_{uv} && \forall (u,v) \in E, \forall i \in \{1,\dots,k\},\\
 & \sum_{(u,v) \in E} x_{uvi} \geq d_{vi} && \forall v \in V \setminus \{s\}, \forall i \in \{1,\dots,k\},\\
 & x_{uvi} \geq y_{uvi} && \forall (u,v) \in E, \forall i \in \{1,\dots,k\},\\
 & \sum_{(u,v) \in E} x_{uvi} \leq \sum_{(u,v) \in E} \phi_{uvi} && \forall  v \in V \setminus \{s\}, \forall i \in \{1,\dots,k\},\\
 & \sum_{(u,v) \in E} x_{uvi} \leq |E| \sum_{(u,v) \in E} y_{uvi} && \forall v \in V \setminus \{s\}, \forall i \in \{1,\dots,k\},\\
 & \sum_{(u,v) \in E} y_{uvi} \leq 1 && \forall v \in V \setminus \{s\}, \forall i \in \{1,\dots,k\}, \\
 & d_{si} = 1 &&  \forall i \in \{1,\dots,k\},\\
 & \phi_{uvi} \leq n y_{uvi} && \forall (u,v) \in E, \forall i \in \{1,\dots,k\},\\
 & \phi_{uvi} \geq -n y_{uvi} && \forall (u,v) \in E, \forall i \in \{1,\dots,k\},\\
 & \phi_{uvi} \leq (d_{vi} - d_{ui}) + (1-y_{uvi})n && \forall (u,v) \in E, \forall i \in \{1,\dots,k\},\\
 & \phi_{uvi} \geq (d_{vi} - d_{ui}) - (1-y_{uvi})n && \forall (u,v) \in E, \forall  i \in \{1,\dots,k\},\\
 & w_i \in \mathbb{Z}^+ && \forall i \in \{1,\dots,k\},\\
 & x_{uvi} \in \{0,1\} && \forall (u,v) \in E, \forall i \in \{1,\dots,k\},\\
 & \pi_{uvi} \in \mathbb{Z}^+ \cup \{0\} && \forall (u,v) \in E, \forall i \in \{1,\dots,k\\
 & d_{vi} \in \mathbb{Z}^+ && \forall v \in V, \forall i \in \{1,\dots,k\},\\
 & y_{uvi} \in \{0,1\} && \forall (u,v) \in E, \forall i \in \{1,\dots,k\},\\
 & \phi_{uvi} \in \mathbb{Z} && \forall (u,v) \in E, \forall i \in \{1,\dots,k\}.
\end{align}
\end{subequations}

\newpage
\section{Full ILP formulation for $k$-Flow Decomposition into Walks ($k$-FDW): $x_{uvi}$ as an integer variables}
\label{sec:k-fdw_II}

\begin{subequations}
\label{mod:pw_II}
\begin{align}
 & \sum_{(s,v) \in E} x_{svi} \leq 1 && \forall i \in \{1, \ldots, k\}, \\
 & \sum_{(u,v) \in E} x_{uvi} - \sum_{(v,w) \in E} x_{vwi} = 0 &&  \forall i \in \{1, \ldots, k\}, \forall v \in V \setminus \{s, t\},\\
 & w_i = \sum_{j \in \{0,\dots,\lfloor\log_2(\overline{w})\rfloor\}} 2^j \zeta_{ij} && \forall i \in \{1, \ldots, k\},\\
 & f_{uv} = \sum_{i \in \{1,\dots,k\}} \sum_{j \in \{0,\dots,\lfloor\log_2(\overline{w})\rfloor\}} 2^j\varphi_{uvij} && \forall (u,v) \in E, \\
 & \varphi_{uvij} \leq f_{uv} \zeta_{ij} && \forall (u,v) \in E, \forall i \in \{1,\dots,k\}, \forall j \in \{0,\dots,\lfloor\log_2(\overline{w})\rfloor\}, \\
 & \varphi_{uvij} \leq x_{uvi} && \forall (u,v) \in E, \forall i \in \{1,\dots,k\}, \forall j \in \{0,\dots,\lfloor\log_2(\overline{w})\rfloor\}, \\
 & \varphi_{uvij} \geq x_{uvi} - (1-\zeta_{ij})f_{uv} && \forall (u,v) \in E, \forall i \in \{1,\dots,k\}, \forall j \in \{0,\dots,\lfloor\log_2(\overline{w})\rfloor\},\\
 & \sum_{(u,v) \in E} x_{uvi} \geq d_{vi} && \forall v \in V \setminus \{s\}, \forall i \in \{1,\dots,k\},\\
 & x_{uvi} \geq y_{uvi} && \forall (u,v) \in E, \forall i \in \{1,\dots,k\},\\
 & \sum_{(u,v) \in E} x_{uvi} \leq \sum_{(u,v) \in E} \phi_{uvi} && \forall  v \in V \setminus \{s\}, \forall i \in \{1,\dots,k\},\\
 & \sum_{(u,v) \in E} x_{uvi} \leq |E| \sum_{(u,v) \in E} y_{uvi} && \forall v \in V \setminus \{s\}, \forall i \in \{1,\dots,k\},\\
 & \sum_{(u,v) \in E} y_{uvi} \leq 1 && \forall v \in V \setminus \{s\}, \forall i \in \{1,\dots,k\}, \\
 & d_{si} = 1 &&  \forall i \in \{1,\dots,k\},\\
 & \phi_{uvi} \leq n y_{uvi} && \forall (u,v) \in E, \forall i \in \{1,\dots,k\},\\
 & \phi_{uvi} \geq -n y_{uvi} && \forall (u,v) \in E, \forall i \in \{1,\dots,k\},\\
 & \phi_{uvi} \leq (d_{vi} - d_{ui}) + (1-y_{uvi})n && \forall (u,v) \in E, \forall i \in \{1,\dots,k\},\\
 & \phi_{uvi} \geq (d_{vi} - d_{ui}) - (1-y_{uvi})n && \forall (u,v) \in E, \forall  i \in \{1,\dots,k\},\\
 & w_i \in \mathbb{Z}^+ && \forall i \in \{1,\dots,k\},\\
 & x_{uvi} \in \mathbb{Z}^+ \cup \{0\} && \forall (u,v) \in E, \forall i \in \{1,\dots,k\},\\
 & d_{vi} \in \mathbb{Z}^+ \cup \{0\} && \forall v \in V, \forall i \in \{1,\dots,k\},\\
 & y_{uvi} \in \{0,1\} && \forall (u,v) \in E, \forall i \in \{1,\dots,k\},\\
 & \zeta_{ij} \in \{0,1\} && \forall i \in \{1,\dots,k\}, \forall j \in \{0,\dots,\lfloor\log_2(\overline{w})\rfloor\}, \\
 & \varphi_{uvij} \in \mathbb{Z}^+ \cup \{0\} && \forall (u,v) \in E, \forall i \in \{1,\dots,k\}, \forall j \in \{0,\dots,\lfloor\log_2(\overline{w})\rfloor\},\\
 & \phi_{uvi} \in \mathbb{Z} && \forall (u,v) \in E, \forall i \in \{1,\dots,k\}.\\
 & \overline{w} = \max_{(u,v) \in E} f_{uv} \notag
\end{align}
\end{subequations}

\end{document}